\documentclass[prx,aps,amsmath,amssymb,floatfix,superscriptaddress,11pt,tightenlines,longbibliography,onecolumn,notitlepage]{revtex4-1}
\usepackage{graphicx}
\usepackage{epstopdf}
\usepackage{amsmath}
\usepackage{amsfonts}
\usepackage{bm}
\usepackage{color}
\usepackage{subfigure}
\usepackage{amsthm}
\usepackage{tikz}
\usepackage{comment}
\usepackage[pdftex,colorlinks=true,linkcolor=blue,citecolor=red,urlcolor=blue]{hyperref}

\newcommand{\ket}[1]{\left|#1\right\rangle}

\newcommand{\Tr}{\textnormal{tr}}
\newcommand{\Supp}[1]{\textnormal{Supp}(#1)}
\newcommand{\tran}[1]{T_{B}^{BS_{#1}}}
\newcommand{\trant}[2]{T_{B}^{BS_{[#1 , #2]}}}
\newcommand{\ttran}[1]{\tilde{T}_{B}^{BS_{#1}}}
\newcommand{\ttrant}[2]{\tilde{T}_{B}^{BS_{[#1 , #2]}}}

\newcommand{\trantd}[2]{T_{B}^{BS_{[#1 , #2]}*}}

\newcommand{\ttrantd}[2]{\tilde{T}_{B}^{BS_{[#1 , #2]}*}}
\newcommand{\U}[2]{\mathcal{U}_{[#1, #2]}}
\newcommand{\tU}[2]{\tilde{\mathcal{U}}_{[#1, #2]}}
\newcommand{\Ohm}[2]{\omega^{S_{[#1, #2]} S_{[#1, #2]}'}}
\newcommand{\tOhm}[2]{\tilde{\omega}^{S_{[#1, #2]} S_{[#1, #2]}'}}

\newtheorem{lem}{Lemma}
\newtheorem{thm}{Theorem}
\newtheorem{defi}{Definition}

\begin{document}
\title{Noise-resilient preparation of quantum many-body ground states}
\author{Isaac H. Kim}
\affiliation{IBM T.J. Watson Research Center}
\date{\today}

\begin{abstract}
  Certain quantum many-body ground states can be prepared by a large-depth quantum circuit consisting of geometrically local gates. In the presence of noise, local expectation values deviate from the correct value at most by an amount comparable to the noise rate. This happens if the action of the noiseless circuit, restricted to certain subsystems, rapidly mixes local observables up to a small correction. The encoding circuit for the surface code is given as an example.
\end{abstract}
\maketitle
\section{Introduction\label{section:Introduction}}
Recently, we proposed a method to simulate a two-dimensional quantum many-body system from a one-dimensional digital quantum simulator\cite{Kim2017a}. A method to find a representation of the ground state wavefunction was proposed. It was shown that ground states of massive systems can be reliably approximated by such a representation. We also gave a heuristic explanation on why the method would be resilient to noise. It was argued that noise rate of $\epsilon$ would only affect local observables by an amount proportional to $\epsilon$. This is a nontrivial claim, because the depth of the underlying circuit scales linearly with the length scale; naively, the error bound would scale with the system size.

The purpose of this paper is to formulate a condition under which such a statement is rigorously true, albeit in a slightly weaker form than what was claimed in Ref.\cite{Kim2017a}. We show that certain ground states can be prepared by a quantum circuit that is \emph{noise-resilient.} This means that, if the circuit obeys a certain condition, the expectation values of local observables are perturbed at most by a function $f(\epsilon)$ which is uniformly bounded in all system size, such that it vanishes in the $\epsilon \to 0$ limit. The choice of this function in this paper is $f(x) = \mathcal{O}(x\log^2 x)$, but of course, other functions with the same property would be able to serve the same purpose.

We embarked on this study to overcome the existing challenges in simulating strongly interacting quantum many-body systems. To name a few, quantum Monte Carlo method suffers from the infamous sign problem, which is unlikely to be resolved\cite{Troyer2005}. While the density matrix renormalization group method\cite{White1992} has enjoyed a tremendous amount of success, its higher-dimensional generalization\cite{Verstraete2004a} cannot be efficiently contracted in general\cite{Schuch2007}. The network can be approximately contracted ``efficiently,'' meaning that the contraction time scales polynomially with the underlying variational parameter\cite{Lubasch2014}. However, the degree of the polynomial is quite large in practice, and this hinders its practical application. Another variant, the multi-scale entanglement renormalization ansatz\cite{Vidal2008}, can be exactly contracted in an efficient manner. However, again the degree of the polynomial is too large to be deemed practical\cite{Evenbly2009}. The main computational bottlenecks behind these methods come from elementary linear algebra operations such as singular value decomposition, which are unlikely to be significantly sped up in near term. While this is not as fundamental as the sign problem, it is important to remember that (i) the memory requirement of these methods scale as a large polynomial of the variational parameter, which can be problematic and (ii) the floating-point operations per second per core has not improved significantly over a prolonged period of time. Without a breakthrough technology that can overcome these issues, further progress seems unlikely.\footnote{However, there might be a more efficient method that can leverage the power of modern graphics processing unit.}

On the quantum side, in principle, a fault-tolerant quantum computer\cite{Shor1996} would be able to overcome these problems\cite{Lloyd1996}. While noise rates below the fault tolerance threshold have been reported\cite{Benhelm2008,Barends2014}, a successful demonstration of fault-tolerance\cite{Gottesman2016} is yet to be seen. To make matters worse, the associated overhead is expected to be large\cite{Fowler2012}. Furthermore, there is another layer of overhead associated to the generation of so called magic states\cite{Bravyi2005}, which adds at least another order of magnitude. All in all, with the existing approaches, to make progress, it seems necessary either to overcome multiple levels of challenges or to develop a breakthrough technology.

Compared to these methods, our proposal\cite{Kim2017a} has a number of advantages. First, because local expectation values are resilient to noise, as long as we only ask questions about such expectation values, we do not need to perform error correction. This subsequently implies that there is no overhead associated to error correction. Indeed, various quantities of interest, e.g., energy per site and local order parameter, are often experimentally measured up to a fixed precision which is independent of the system size. If our goal is to simulate such properties of the physical system, a constant precision would be often good enough. Second, as explained in Ref.\cite{Kim2017a}, the proposed one-dimensional architecture already exists\cite{Kim2010,Barends2014}, and there does not seem to be much fundamental difficulty in scaling it up. Third, on the software side, the time to measure local quantities, e.g., local magnetization, scales as $\mathcal{O}(\ell D /\delta^2)$, where $\ell$ is the length scale of the underlying system, $D$ is the depth of the circuit, and $\delta$ is the desired precision; here we expect $D=\mathcal{O}(1)$ for physical systems of interest. The linear dependence on $\ell$ is not only efficient in a complexity-theoretic sense, but also in a practical sense. The factor of $1/\delta^2$ can be nonnegligible, but this comes from a statistical sampling error, which can be removed by running the simulation in parallel with many devices. We believe these qualities make our proposal a viable option for simulating strongly interacting quantum many-body systems in near-term quantum devices.

While this paper was motivated from the practical problem of simulating quantum many-body systems, from the perspective of quantum error correction, it is amusing to note that there is a nontrivial method to bound the effect of error even when the number of circuit elements that influence the relevant observables diverges. The analysis, by definition, has to be completely different from the error analysis in fault tolerant quantum computation, wherein the effect of noise is bounded by counting and weighting different fault paths\cite{Aharonov1996}. This is an interesting phenomenon on its own, and certainly unexpected from the perspective of quantum error correction. What made this possible? The short answer would be that the physical nature of the simulated system makes the simulation noise-resilient. The condition that we formulated seems mild enough that it seems to be applicable to many physical systems of interest. Indeed, we list a few examples in Section \ref{section:examples}.

This condition, which concerns the stability of dissipative discrete-time local dynamics, was inspired from an interesting recent work of Lucia et al.\cite{Cubitt2013,Lucia2014}. They were able to prove that local expectation values of the steady state of a dissipative system is perturbed at most by an amount comparable to the perturbation strength, provided that the underlying (unperturbed) dynamics rapidly equilibrates local observables. This result is noteworthy; it is normally difficult to bound the effect of perturbation at long-time scale in a controlled manner, yet their work precisely achieves that. Modulo some minor subtleties, their argument was general enough that it could be applied to our work.

The observation that the stability of dissipative system can be used to prove the noise-resilience of large-depth quantum simulation suggests that there may be circumstances in which error correction becomes unnecessary. This shows that we can also make progress by looking for physical properties that can be reliably replicated in a noisy quantum device, instead of merely trying to reduce the noise by error correction. We hope this work initiates further studies in this direction. 

The main result of this paper is summarized in Section \ref{section:summary}. In Section \ref{section:ref}, we discuss notations and facts that are frequently used in this paper. Through Section \ref{section:Heisenberg} and Section \ref{section:local_uniform_contraction}, we derive a bound on the difference between the local expectation values of noiseless and noisy quantum circuit. In Section \ref{section:examples} we discuss nontrivial applications of this bound. We close with several open questions in Section \ref{section:Discussion}.  

\section{Summary\label{section:summary}}
As in the conventional circuit model of quantum computation, our proposal\cite{Kim2017a} involves quantum gates, measurements, and state preparations. None of these processes can be assumed to be perfect. We assume that they are perturbed by a noise of strength  $\epsilon$. Up to a polylogarithmic factor in $\epsilon$, the noise strength bounds the deviation of noisy local expectation values from their true values.

This bound is derived for certain circuits that can be implemented in holographic quantum simulation\cite{Kim2017a}. We first derive a bound that holds for any such circuit, and then apply the bound to certain examples. The main result of this paper is Theorem \ref{thm:local_rapid_mixing_stability}, which is explained in Section \ref{section:main} in colloquial terms. Below, we begin by reviewing the structure of circuits that can be implemented in holographic quantum simulation. We then formalize our assumptions on noise, state the main result, and sketch the proof.

\subsection{Holographic Quantum Simulation\label{section:HQS}}
Recently, we proposed a method to simulate a two-dimensional(2D) system from a one-dimensional(1D) quantum device. The qubits of the 1D device are classified into the bath, sink, and the system qubits; see FIG.\ref{fig:architecture}. The state of the 2D system is represented by a history of the 1D device:
\begin{equation}
  \langle O_{(x_1,y_1)}\cdots O_{(x_n,y_n)}\rangle_{\text{2D}} := \langle O_{x_1}(y_1) \cdots O_{x_n}(y_n)\rangle_{\text{1D}}, \label{eq:1D_2D}
\end{equation}
where $\langle \cdots \rangle_{\text{2D}}$ is the expectation value of the observables over the state of the 2D system and $\langle \cdots \rangle_{\text{1D}}$ is the (time-dependent) expectation value of the observables over the history of the 1D system. Here $O_{(x_1,y_1)}$ is an observable on the qubit at location $(x,y)$ of the 2D system and $O_x(y)$ is an observable on the $(x)$th system qubit of the 1D device at time $y$.
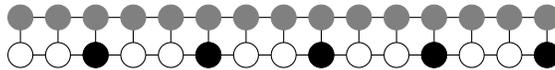
\begin{figure}[h]
  \begin{tikzpicture}
    \draw (0.5,0.5)--(7.5,0.5);
    \draw (0.5,0)--(7.5,0);
    \foreach \i in {1,...,15}
        {\draw (0.5*\i,0.0)--(0.5*\i,0.5);}
    \foreach \i in {1,...,15}
        {\node[circle,draw=gray,fill=gray](\i) at (0.5 * \i, 0.5) {};}
    \foreach \i in {1,...,15}
        {\node[circle,draw=black,fill=white](\i) at (0.5 * \i, 0.0) {};}            
    \foreach \i in {1,...,5}
        {\node[circle,draw=black,fill=black](\i) at (1.5 * \i, 0.0) {};} 
  \end{tikzpicture}
  \caption{Locality structure of the physical system that implements holographic quantum simulation. The two-qubit gates are implemented over qubits that are connected by an edge. The black qubits are the system qubits, the white qubits are the bath qubits, and the gray qubits are the sink qubits. \label{fig:architecture}}
\end{figure}

At each time step, the system and the sink qubits are initialized, a low-depth local quantum circuit consisting of nearest-neighbor gates are applied, and if necessary, the system qubits are measured. The circuit builds correlation between the bath and the system qubit, which subsequently builds correlation between system qubits at different time; in effect, the bath mediates the correlation between system qubits at different times. The initialization ensures that Eq.\ref{eq:1D_2D} defines a valid quantum mechanical state.

Intuitively, one can view this process as a 1D system ``gliding'' over a 2D system; see FIG.\ref{fig:gliding}. The 2D system, which can be thought as a collection of system qubits at different times, are initialized to some fixed product state. The 1D system, which can be thought as the bath-sink composite, glides over the 2D system and sequentially interacts with different rows of system qubits. After the 1D system reaches the last row, a desired 2D state is prepared. While one can certainly perform measurements at this point, it is possible to perform arbitrary measurements while keeping just one row of system qubits; see Ref.\cite{Kim2017a} for detail.
\begin{figure}[h]
  \begin{tikzpicture}[xscale=0.5,yscale=0.4]
    \draw[<->] (1,-0.5,15)--(7,-0.5,15) node [midway, below] {$\ell_x$};
    \foreach \i in {1}
        \foreach \j in {1,...,15}
                {\draw (\i,0.5,\j)--(\i,0,\j);}
    \draw (1,0.5,1)--(1,0.5,15);
    \foreach \i in {1,...,7}            
        {\draw[dashed] (\i,0,3)--(\i,0,15);}
    \foreach \i in {1,...,5}
        {\draw[dashed] (1,0,3*\i)--(7,0,3*\i);}    
    \foreach \i in {1}
        {\draw (\i,0,1)--(\i,0,15);}
    \foreach \i in {1}
        \foreach \j in {1,...,15}
                {\node[circle,draw=black,fill=white,scale=0.4](\i) at (\i, 0, \j) {};}
    \foreach \i in {1,...,7}
        \foreach \j in {1,...,5}
                {\node[circle,draw=black,fill=black,scale=0.4](\i) at (\i,0, 3*\j) {};}
    \foreach \j in {1,...,15}
             {\node[circle,draw=gray,fill=gray,scale=0.4](\j) at (1,0.5,\j) {};}
    \draw[->] (7.75,0,7.5)--(8.25,0,7.5);
    \begin{scope}[xshift=8cm]
        \foreach \i in {2}
        \foreach \j in {1,...,15}
                {\draw (2,0.5,\j)--(\i,0,\j);}
    \foreach \i in {1,...,7}            
        {\draw[dashed] (\i,0,3)--(\i,0,15);}
    \foreach \i in {1,...,5}
        {\draw[dashed] (1,0,3*\i)--(7,0,3*\i);}
    \draw (2,0.5,1)--(2,0.5,15);
    \foreach \i in {2}
        {\draw (\i,0,1)--(\i,0,15);}
    \foreach \i in {2}
        \foreach \j in {1,...,15}
                {\node[circle,draw=black,fill=white,scale=0.4](\i) at (\i, 0, \j) {};}
    \foreach \i in {1,...,7}
        \foreach \j in {1,...,5}
                {\node[circle,draw=black,fill=black,scale=0.4](\i) at (\i,0, 3*\j) {};}
    \foreach \j in {1,...,15}
             {\node[circle,draw=gray,fill=gray,scale=0.4](\j) at (2,0.5,\j) {};}
    \draw[<->] (7.25,-0.25,3)--(7.25,-0.25,15) node [midway, right, below] {$\ell_y$};
    \end{scope}            
  \end{tikzpicture}
\caption{A 1D system is gliding over a 2D system consisting of $\ell_x \times \ell_y=7 \times 5$ qubits. Nearby qubits interact with each other. \label{fig:gliding}}
\end{figure}
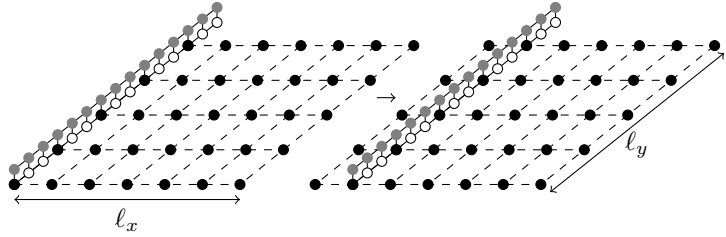

Formally speaking, the 2D state is a finitely correlated state(FCS)\cite{Fannes1992}. In modern language, such state can be represented by a matrix product density operator\cite{Verstraete2004,Zwolak2004}. While the matrix product density operator is certainly a more well-known formalism, we chose to use the language of FCS for the following reason. The matrices that define the matrix product density operator implicitly define a family of completely positive trace-preserving maps\cite{Verstraete2004}, i.e., quantum channels. In holographic quantum simulation, these channels are not some mathematical object, but rather an actual physical operation implemented in the experiment. We are interested in understanding how the experiment is affected by imperfect implementation of these channels. These channels, in our proposal, have a number of properties that we extensively use. We are unaware of any method to translate these properties into the properties of the matrices that appear in the definition of the matrix product density operator. On the other hand, FCS is formulated only in terms of channels. This is why the language of FCS is more appropriate for our analysis.

A FCS is defined in terms of the initial state of the bath($B$) and channels that involve the bath and the system at different times; the system at time $t$ is denoted as $S_t$. We shall use $S_t'$ to denote the sink at time $t$. For qubits arranged on a $\ell_x \times \ell_y$ grid, the state has the following form:
\begin{equation}
\rho = \Tr_B[T_B^{BS_{\ell_y}} \circ \cdots \circ T_B^{BS_1}(\rho^B)], \label{eq:state}
\end{equation}
where $\rho^B$ is a state of the bath, $T_B^{BS_t}$ is a \emph{transition map}, a channel from the bath to the composite of bath and system at time $t$, and $\Tr_B$ is a partial trace over the bath qubits. Here, the identity superoperators are suppressed. For example, $T_B^{BS_{\ell_y}}$ is a short-hand notation for $T_B^{BS_{\ell_y}} \otimes \mathcal{I}_{S_1\cdots S_{\ell_y-1}}$, where $\mathcal{I}_{S_1\cdots S_{\ell_y-1}}$ is the identity superoperator over $S_1,S_2,\cdots, S_{\ell_y-1}$. Physically, these channels are implemented by initializing the system and the sink to a fixed state, and then applying a circuit to the physical system. This means that
\begin{equation}
  T_B^{BS_t}(\cdot) = \Tr_{S_t'}[\mathcal{U}_t((\cdot) \otimes \omega^{S_t'} \otimes \omega^{S_t})],
  \label{eq:def1}
\end{equation}
where  $\mathcal{U}_t$ is a local depth-$D$ quantum circuit applied to $BS_t'S_t$, and $\omega_{S_t}$ is a fixed product state. That is,
\begin{equation}
\mathcal{U}_t = \mathcal{U}_t^{(D)} \circ \cdots \circ \mathcal{U}_t^{(1)},\label{eq:def2}
\end{equation}
where
\begin{equation}
\mathcal{U}_t^{(i)} = \bigotimes_j \mathcal{U}_t^{(i,j)}.\label{eq:def3}
\end{equation}
Here $\mathcal{U}_{t}^{(i,j)}$ is a nearest-neighbor two qubit gate; see FIG.\ref{fig:architecture}.

All the initial states are assumed to be a product state. That is,
\begin{equation}
  \begin{aligned}
    \rho^{B} &= \bigotimes_{i\in B} \rho^i \\
    \omega^{S_t} &= \bigotimes_{i\in S_t} \omega^i \\
    \omega^{S_t'} &= \bigotimes_{i\in S_t'} \omega^i.
  \end{aligned}
\end{equation}

Once the state is prepared, relevant observables can be measured. Since any observable can be decomposed into a linear combination of Pauli operators, it suffices to define the measurement protocol for these operators. In our proposal, for a Pauli operator over $n$ qubits, each of the qubits are measured in the eigenbasis of the Pauli operators. The average of the product of the measured values is the expectation value.

\subsection{Assumptions on noise\label{section:noise}}
 The noisy state, which shall be denoted as $\tilde{\rho}$ throughout this paper, is defined by replacing the gates and state preparations by their noisy counterparts. Specifically,
\begin{equation}
\tilde{\rho} = \Tr_B[\tilde{T}_B^{BS_{\ell_y}} \circ \cdots \circ \tilde{T}_{B}^{BS_1}(\tilde{\rho}^B)],
\end{equation}
where
\begin{equation}
\tilde{T}_B^{BS_t}(\cdot) = \Tr_{S_t'}[\mathcal{\tilde{U}}_t((\cdot) \otimes \tilde{\omega}^{S_t'} \otimes \tilde{\omega}^{S_t})].
\end{equation}
The noisy circuit($\mathcal{\tilde{U}}_t$) is related to the noiseless circuit($\mathcal{U}_t$) by the following relation:
\begin{equation}
  \begin{aligned}
    \tilde{\mathcal{U}}_t = \tilde{\mathcal{U}}_t^{(D)} \circ \cdots \circ \tilde{\mathcal{U}}_t^{(1)}\\
    \tilde{\mathcal{U}}_t^{(i)} = \bigotimes_{j} \tilde{\mathcal{U}}_t^{(i,j)},
  \end{aligned}
\end{equation}
where $\|\tilde{\mathcal{U}}_{t}^{(i,j)} - \mathcal{U}_t^{(i,j)} \|_{\diamond} \leq \epsilon$. Here $\|\cdots \|_{\diamond}$ is the so called diamond norm, also known as the completely-bounded norm\cite{Kitaev2002}.

The noisy initial states are assumed to be
\begin{equation}
  \begin{aligned}
    \tilde{\rho}^B &= \bigotimes_{i\in B} \tilde{\rho}^i\\
    \tilde{\omega}^{S_t} &= \bigotimes_{i\in S_t} \tilde{\omega}^i \\
    \tilde{\omega}^{S_t'} &= \bigotimes_{i\in S_t'} \tilde{\omega}^i,
    \end{aligned}
\end{equation}
where $\|\tilde{\rho}^i -\rho^i \|_1, \|\tilde{\omega}^i - \omega^i\|_1 \leq \epsilon$. Here $\| \cdots \|_1$ is the trace norm.

A noisy Pauli measurement is modeled by replacing a tensor product of Pauli operators into a tensor product of noisy Pauli operators. Without loss of generality, consider the following Pauli operator:
\begin{equation}
O = \bigotimes_{i=1}^n \sigma^{a_i},
\end{equation}
where $a_i\in \{I,x,y,z \}$ is the index of the Pauli operator. A noisy measurement of this operator is modeled by a noiseless measurement of the following operator:
\begin{equation}
\tilde{O} = \bigotimes_{i=1}^{n} \tilde{\sigma}^{a_i},
\end{equation}
where $\|\sigma^{a_i} - \tilde{\sigma}^{a_i} \|\leq \epsilon$ and $\|\tilde{\sigma}^{a_i} \|\leq 1$. Here $\|\cdots \|$ is the operator norm.

\subsection{Main result\label{section:main}}
The main result of this paper can be divided into two parts. The first part is a bound that holds for any quantum circuit that can be implemented in holographic quantum simulation; see Theorem \ref{thm:local_rapid_mixing_stability}. While we have not introduced all the relevant definitions yet, the content of this theorem can be summarized as follows. Consider the restriction of the transition map to the bath. If any local observable in the bath equilibrates exponentially fast in time, then local expectation values are stable to noise. 

In the second part, we apply Theorem \ref{thm:local_rapid_mixing_stability} to specific examples: the ground state of the surface code and the trivial state. Applied to these circuits, we show that
\begin{equation}
  |\Tr[\rho O - \tilde{\rho}\tilde{O}]| \leq \mathcal{O}(\|O \| \epsilon \log^2\left(1/\epsilon\right)),\label{eq:main}
\end{equation}
uniformly in $\ell_x,\ell_y$, where $O$ is an operator supported on a ball of bounded radius.

It is important to note that, while the bound concerns two different states, the main result is really a statement about the underlying circuits; these circuits are described in Section \ref{section:examples}. After all, the state $\tilde{\rho}$ is defined in terms of operations that are noisy versions of the circuit that prepares $\rho$.

There is another important subtlety: a different circuit may be able to prepare the same state, yet does not have the same kind of noise resilience. This is echoing once more that the main result concerns the circuit, not the state. Indeed, we provide explicit circuits in Section \ref{section:examples}, but they are by no means the only circuits that can prepare those states. Our main result only concerns the circuits we describe in Section \ref{section:examples}.

\subsection{Structure of the proof\label{section:proof_sketch}}
The argument consists of two parts; the first part is independent of the details of the circuit, whereas the second part crucially depends on these details. Section \ref{section:Heisenberg} and \ref{section:local_uniform_contraction} covers the first part. In Section \ref{section:Heisenberg}, we derive an upper bound on $|\Tr[\rho O - \tilde{\rho}\tilde{O}]|$ in terms of a difference between the expectation values of local observables undergoing two different dynamics: a time-dependent discrete-time local dissipative dynamics and its noisy counterpart. In Section \ref{section:local_uniform_contraction}, this difference is bounded in terms of a function that only depends on the circuit that induces the noiseless dynamics. This function, applied to arbitrary circuits, is likely to yield a lousy bound. However, for the circuits discussed in this paper, the bound yields the main result. This analysis, which is the main content of Section \ref{section:examples}, is model-specific, and constitutes the second part of our argument.

\section{Notations and facts\label{section:ref}}
We list a set of conventions and facts; they shall be referenced many times in the proof. In the first reading, the readers can skip the content of Section \ref{section:states}-\ref{section:observables} and intermittently come back when the relevant content is referred from the proof. 

Here is our convention on the Hilbert space. We shall formally treat $S_t$ and $S_t'$ to be different Hilbert spaces for all $t$. The qubits within each $S_t$ and $S_t'$ shall be referred to as $i_t$ and $i_t'$, where $i$ is an integer ranging from $1$ to $\ell_x$. Also,
\begin{equation}
  \begin{aligned}
    S_{[t,t']} &= \bigcup_{n\in [t,t']} S_n \\
    S_{[t,t']}' &= \bigcup_{n\in [t,t']} S_n',
  \end{aligned}  
\end{equation}
and
\begin{equation}
  \begin{aligned}
    \mathcal{H}_{S_{[t,t']}} &= \bigotimes_{n\in [t,t']} \mathcal{H}_{S_n} \\
    \mathcal{H}_{S_{[t,t']}'} &= \bigotimes_{n\in [t,t']} \mathcal{H}_{S_n'},
  \end{aligned}
\end{equation}
where $[t,t'] = \{n \in \mathbb{Z}, t\leq n\leq t'  \}.$

Except for $\rho$ and $\tilde{\rho}$, which are states acting on $\mathcal{H}_{S_{[1,\ell_y]}}$, all the states will be accompanied with their domain in the superscript. Indeed, this is the convention we have used so far in Section \ref{section:summary}. On top of the states that are already defined, we shall use the following set of states:
\begin{equation}
  \begin{aligned}
  \omega^{S_{t_1}S_{t_2}'} &= \omega^{S_{t_1}} \otimes \omega^{S_{t_2}'} \\
  \tilde{\omega}^{S_{t_1}S_{t_2}'} &= \tilde{\omega}^{S_{t_1}} \otimes \tilde{\omega}^{S_{t_2}'}
  \end{aligned}
\end{equation}
Also,
\begin{equation}
  \begin{aligned}
    \omega^{S_{[t_1,t_1']}S_{[t_2,t_2']}'} &= \bigotimes_{n\in [t_1,t_1']} \omega^{S_n}  \bigotimes_{m\in [t_2,t_2']} \omega^{S_m'}\\
    \tilde{\omega}^{S_{[t_1,t_1']}S_{[t_2,t_2']'}} &= \bigotimes_{n\in [t_1,t_1']} \tilde{\omega}^{S_n} \bigotimes_{m\in [t_2,t_2']} \tilde{\omega}^{S_{m}'}.
  \end{aligned}
\end{equation}

For quantum channels, we shall denote the domain and the codomain in their subscript and superscript; an exception to this rule is the unitary operators($\mathcal{U}_t$) and their noisy counterparts($\tilde{\mathcal{U}}_t$). A dual of a channel with respect to the Hilbert-Schmidt inner product is denoted by placing a $*$ in the superscript. The composition of the aforementioned channels($T_B^{BS_t}$ and $\tilde{T}_B^{BS_t}$) shall appear frequently:
\begin{equation}
  \begin{aligned}
    \trant{t}{t'} &=  \tran{t'} \circ \tran{t'-1} \circ \cdots \circ\tran{t+1} \circ \tran{t} \\
    \ttrant{t}{t'} &=  \ttran{t'} \circ \ttran{t'-1} \circ \cdots \circ\ttran{t+1} \circ \ttran{t}.
  \end{aligned}
\end{equation}
 Also,
\begin{equation}
  \begin{aligned}
    \U{t}{t'} &= \mathcal{U}_{t'} \circ \mathcal{U}_{t'-1} \circ \cdots \circ \mathcal{U}_{t+1} \circ \mathcal{U}_{t} \\
     \tU{t}{t'} &= \tilde{\mathcal{U}}_{t'} \circ \tilde{\mathcal{U}}_{t'-1} \circ \cdots \circ \tilde{\mathcal{U}}_{t+1} \circ \tilde{\mathcal{U}}_{t}
  \end{aligned}
\end{equation}
Let us again emphasize that the identity superoperators are suppressed in the above expressions. Unless specified otherwise, this will be the convention we employ.

The support of the operators shall be generally suppressed. When necessary, the nontrivial support of an operator $O$ shall be denoted as $\text{Supp}(O)$; we define $|\text{Supp}(O)|$ as its volume. A set of operators acting on a Hilbert space $\mathcal{H}$ shall be denoted as $\mathcal{B}(\mathcal{H})$. Lastly, we shall suppress the proportionality constant by using the physics convention of $\mathcal{O}$ notation. That is, $\mathcal{O}(f(x))$ is a function that is bounded between $cf(x)$ and $c'f(x)$ for some constants $c$ and $c'$.

\subsection{Locality \label{section:locality}}
As we explained in Section \ref{section:HQS}, the state $\rho$ can be thought as a state created by  gliding a 1D system over a 2D system and letting them interact with each other. In the intermediate stages, the joint state of the 1D and the 2D system can be written as follows:
\begin{equation}
  \begin{aligned}
    \trant{1}{t}&(\rho^B) \\
    \ttrant{1}{t}&(\tilde{\rho}^B),
    \end{aligned}
\end{equation}
where $t\in [1,\ell_y]$.

It is important to note that states at different times($t$) act on different Hilbert spaces. In particular, depending on the time, what it means for an operator to be ``local'' changes. Our analysis crucially relies on understanding the behavior of these operators, and as such, we must define the notion of distance for different times.

Distance between two different qubits at time $t$ is defined in terms of the interaction graph, $G^{(t)}=(V^{(t)}, E^{(t)})$. This graph can be thought as a merger between a graph that represents the locality structure of the simulated system and the locality structure of the physical device. Specifically,
\begin{equation}
  \label{eq:locality}
  \begin{aligned}
    V^{(t)} &= V_{\text{sim}}^{(t)} \cup V_{\text{dev}}^{(t)} \\
    E^{(t)} &= E_{\text{sim}}^{(t)} \cup E_{\text{dev}}^{(t)}, 
  \end{aligned}
\end{equation}
where $G_{\text{sim}}^{(t)} = (V_{\text{sim}}^{(t)}, E_{\text{sim}}^{(t)})$ is the graph for the simulated system and $G_{\text{dev}}^{(t)} =(V_{\text{dev}}^{(t)}, E_{\text{dev}}^{(t)}) $ is the graph for the physical device. The former is defined as
\begin{gather*}
    V_{\text{sim}}^{(t)} = \bigcup_{n\in [1,t]} S_n \\
    E_{\text{sim}}^{(t)} = \{(i_n,j_n)| n\in [1,t], |i-j|=1 \} \cup \{(i_n,i_m)| n,m\in [1,t], |n-m|=1,  \}.
  \end{gather*}
For example, see FIG.\ref{fig:locality_sim}. The latter depends on the details of the arrangement of bath, sink, and system qubits. Our analysis is independent of these details as long as they are arranged in a regular pattern and the ratio between the number of system qubits and the rest is bounded. One such example would be FIG.\ref{fig:architecture}. The merged graph is described in FIG.\ref{fig:locality}(cf. Eq.\ref{eq:locality}).
\begin{figure}[h]
  \begin{tikzpicture}[scale=0.5]
    \foreach \x in {1,...,5}
    \foreach \y in {1,...,5}
             {
         \node[circle,draw=black,fill=black,scale=0.6](\x) at (\x, \y) {};
             }
    \foreach \x in {1,...,5}
        {\draw (\x,1)--(\x,5);}
    \foreach \y in {1,...,5}
             {\draw (1,\y)--(5,\y);}
    \draw[->] (0.9,0.5)--(5.2,0.5) node [midway, below] {$t$};
  \end{tikzpicture}
  \caption{An example of $G_{\text{sim}}^{(t)}$. Here $\ell_x=5$ and $t=5$.  \label{fig:locality_sim}}
\end{figure}
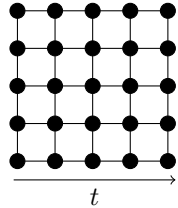
\begin{figure}[h]
  \begin{tikzpicture}[scale=0.5]
        \draw[->] (1.3,0,16)--(7.5,0,16) node [midway, below] {$t$};
    \foreach \i in {5}
        \foreach \j in {1,...,15}
                {\draw (\i,0.5,\j)--(\i,0,\j);}
    \draw (5,0.5,1)--(5,0.5,15);
    \draw (5,0,1)--(5,0,15);
    \foreach \i in {1,...,5}
             {\draw (1,0,3*\i)--(5,0,3*\i);
               \draw[dashed] (5,0,3*\i)--(7,0,3*\i);
             }
    \foreach \i in {1,...,5}
             {\draw (\i,0,3)--(\i,0,15);}
    \foreach \i in {6,7}
             {\draw[dashed] (\i,0,3)--(\i,0,15);}
    \foreach \i in {5}
        \foreach \j in {1,...,15}
                {\node[circle,draw=black,fill=white,scale=0.4](\i) at (\i, 0, \j) {};}
    \foreach \i in {1,...,7}
        \foreach \j in {1,...,5}
                {\node[circle,draw=black,fill=black,scale=0.4](\i) at (\i,0, 3*\j) {};}
    \foreach \j in {1,...,15}
             {\node[circle,draw=gray,fill=gray,scale=0.4](\j) at (5,0.5,\j) {};}
  \end{tikzpicture}
  \caption{The interaction graph $G^{(t)}$ at time $t$ is defined as a merger between $G_{\text{sim}}^{(t)}$(cf. FIG.\ref{fig:locality_sim}) and $G_{\text{dev}}^{(t)}$(cf. FIG.\ref{fig:architecture}). Two qubits are their neighbors if they are connected by an edge. Distance is defined as the graph distance. This figure describes the interaction graph for $t=5$, where $\ell_x=7$ and $\ell_y=5$. The dashed lines are \emph{not} included in the edges of the interaction graph. They are merely a bookkeeping device.\label{fig:locality}}
\end{figure}
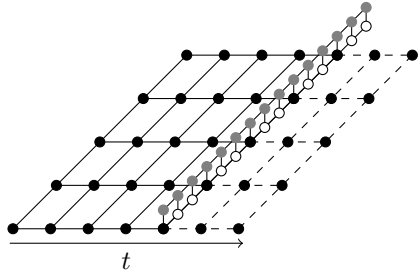

A distance at time $t$ between two qubits is defined as the graph distance between the corresponding vertices in $G^{(t)}$. An operator $O$ at time $t$ is said to be a local if it is supported on a ball of radius $r=\mathcal{O}(1)$ in $G^{(t)}$.

\subsection{States\label{section:states}}
We bound the difference between the local expectation values of noiseless states$(\rho^B$ and  $\omega^{S_{[t_1,t_1']}S_{[t_2,t_2']}'})$ and the noisy states$(\tilde{\rho}^B$ and $\tilde{\omega}^{S_{[t_1,t_1']}S_{[t_2,t_2']}'})$.
\begin{lem}
  For $O\in\mathcal{B}(\mathcal{H}_B)$
  \label{lemma:state_bath}
  \begin{equation}
    |\Tr[(\rho^B - \tilde{\rho}^B)O ]| \leq \epsilon \|O \| |\textnormal{Supp}(O)|
  \end{equation}
\end{lem}
\begin{proof}
  Without loss of generality, let $\text{Supp}(O)=\{a_1,\cdots, a_n \}$ and $X=\text{Supp}(O)$.
  \begin{equation}
    \begin{aligned}
      |\Tr[(\rho^B - \tilde{\rho}^B)O]| &=|\Tr[(\rho^X -  \tilde{\rho}^X ) O]| \\
      &\leq |\Tr[(\rho^{a_1}-\tilde{\rho}^{a_1})\otimes \rho^{X\setminus \{a_1 \}} O]|  + |\Tr[\tilde{\rho}^{a_1}\otimes(\rho^{X\setminus \{a_1 \}} - \tilde{\rho}^{X\setminus \{a_1 \}})O]| \\
      &\leq |\Tr[(\rho^{a_1} - \tilde{\rho}^{a_1})O']| + |\Tr[(\rho^{X\setminus \{a_1 \}} - \tilde{\rho}^{X\setminus \{a_1 \}})O'']|,
    \end{aligned}
  \end{equation}
  where $O'= \Tr_{X\setminus \{a_1\}}[\rho^{X\setminus\{a_1 \}}O]$ and $O'' = \Tr_{\{a_1\}}[\tilde{\rho}^{a_1} O]$. Note that
  \begin{equation}
    \begin{aligned}
      \|O' \| &=\sup_{\substack{\sigma^{a_1}\geq 0 \\ \Tr(\sigma^{a_1})=1}} |\Tr[\sigma^{a_1} \otimes \rho^{X\setminus a_1}O]|\\
      &\leq \sup_{\substack{\sigma^X\geq 0 \\ \Tr(\sigma)=1}} |\Tr[\sigma^X O]| \\
      &=\|O \|.
    \end{aligned}
  \end{equation}
Similarly, $\|O''\| \leq \| O \|$. Since $|\Tr[(\rho^{a_1}-\tilde{\rho}^{a_1})O']\leq \|\rho^{a_1} - \tilde{\rho}^{a_1}\|_1 \| O'\|$
\begin{equation}
    \sup_{\substack{\|O \|\leq 1\\ \text{Supp}(O) \subset X}}|\Tr[(\rho^X - \tilde{\rho}^X)O]| \leq \|O\|\epsilon + \sup_{\substack{\|O \|\leq 1 \\ \text{Supp}(O)\subset X\setminus \{a_1\} }} |\Tr[(\rho^{X\setminus a_1} - \tilde{\rho}^{X\setminus a_1}) O]|.
\end{equation}
The claim is proved by iterating the bound $|\textnormal{Supp}(O)|$ times.
 \end{proof}

A similar bound can be derived for the states of the system/sink.
\begin{lem}
  \label{lemma:state_systemsink}
  For $O \in \mathcal{B}(\mathcal{H})$, where $\mathcal{H} = \mathcal{H}_{S_{[t_1,t_1']} S_{[t_2,t_2']}'}$, \label{eq:state_systemsink}
  \begin{equation}
    \|\Tr_{S_{[t_1,t_1']}S_{[t_2,t_2']}'}[(\omega^{S_{[t_1,t_1']}S_{[t_2,t_2']}'} - \tilde{\omega}^{S_{[t_1,t_1']}S_{[t_2,t_2']}'} ) O ]\| \leq \epsilon \|O \| |\Supp{O}|.
  \end{equation} 
\end{lem}
\begin{proof}
  Without loss of generality, let $\text{Supp}(O)=\{a_1,\cdots, a_n \}$ and $X=\text{Supp}(O)$. Then the left hand side of Eq.\ref{eq:state_systemsink} is equal to $ \|\Tr_{S_{[t_1,t_1']}S_{[t_2,t_2']}'}[(\omega^X-\tilde{\omega}^X  ) O ]\|$. This is bounded as
  \begin{equation}
    \begin{aligned}
      \|\Tr_{S_{[t_1,t_1']}S_{[t_2,t_2']}'}[(\omega^X-\tilde{\omega}^X  ) O ]\| &= \sup_{\substack{\sigma\geq 0 \\ \Tr[\sigma]=1}}|\Tr[\sigma\otimes (\omega^X-\tilde{\omega}^X)O]| \\
      &\leq \delta_1+\delta_2,
    \end{aligned}
  \end{equation}
  \begin{equation}
  \begin{aligned}
    \delta_1&=  \sup_{\substack{\sigma\geq 0 \\ \Tr[\sigma]=1}} |\Tr[\sigma\otimes (\omega^{a_1}-\tilde{\omega}^{a_1})\otimes \omega^{X\setminus \{a_1 \}} O]|  \\
    \delta_2&= \sup_{\substack{\sigma\geq 0 \\ \Tr[\sigma]=1}} |\Tr[\sigma \otimes \tilde{\omega}^{a_1}\otimes(\omega^{X\setminus \{a_1 \}} - \tilde{\omega}^{X\setminus \{a_1 \}})O]|,
  \end{aligned}
  \end{equation}
  where $\sigma\in \mathcal{B}(\mathcal{H}')$. The first term can be bounded as
  \begin{equation}
    \begin{aligned}
      \delta_1 &=  \sup_{\substack{\sigma\geq 0 \\ \Tr[\sigma]=1}}|\Tr[\sigma\otimes (\omega^{a_1} - \tilde{\omega}^{a_1})O']| \\
      &\leq \epsilon\|O\|,
    \end{aligned}
  \end{equation}
  where  $O'= \Tr_{X\setminus \{a_1\}}[\omega^{X\setminus\{a_1 \}}O]$. The second term can be bounded as
  \begin{equation}
    \begin{aligned}
      \delta_2 &=  \sup_{\substack{\sigma\geq 0 \\ \Tr[\sigma]=1}}|\Tr[\sigma\otimes(\omega^{X\setminus \{a_1 \}} - \tilde{\omega}^{X\setminus \{a_1 \}})O'']| \\
        &\leq \|\Tr_{S_{[t_1,t_1']}S_{[t_2,t_2']}' }[(\omega^{X\setminus a_1} - \tilde{\omega}^{X\setminus a_1}) O]\|
    \end{aligned}
  \end{equation}
where $O'' = \Tr_{\{a_1\}}[\tilde{\omega}^{a_1} O]$. In short,
\begin{equation}
    \sup_{\substack{\|O \|\leq 1\\ \text{Supp}(O) \subset X}}\|\Tr_{S_{[t_1,t_1']}S_{[t_2,t_2']}'} [(\omega^X - \tilde{\omega}^X)O]\| \leq \epsilon + \sup_{\substack{\|O \|\leq 1 \\ \text{Supp}(O)\subset X\setminus \{a_1\} }} \|\Tr_{S_{[t_1,t_1']}S_{[t_2,t_2']}'}[(\omega^{X\setminus a_1} - \tilde{\omega}^{X\setminus a_1}) O]\|.
  \end{equation}
Iterating the bound $|\textnormal{Supp}(O)|$ times, the claim is proved.
\end{proof}

\subsection{Channels\label{section:processes}}
We study the properties of the channels that define the state $\rho$ and $\tilde{\rho}$, i.e., $T_B^{BS_t}$ and $\tilde{T}_B^{BS_t}$. Let us begin with their algebraic properties.
\begin{lem}
  \label{lemma:dual}
  \begin{equation}
    \begin{aligned}
      \trantd{t}{t'}(\cdot) &= \Tr_{S_{[t,t']} S_{[t,t']}'}[\Ohm{t}{t'} \U{t}{t'}^* (I_{S_{[t,t']'}} \otimes (\cdot))]\\  
      \ttrantd{t}{t'}(\cdot) &= \Tr_{S_{[t,t']} S_{[t,t']}'}[\tOhm{t}{t'} \tU{t}{t'}^* (I_{S_{[t,t']'}} \otimes (\cdot))]. 
    \end{aligned}
  \end{equation} 
\end{lem}
\begin{proof}
  Recall the definition of $\tran{t}$ and $\ttran{t}$:
  \begin{equation}
    \begin{aligned}
      \tran{t}(\cdot) &= \Tr_{S_t'}[\mathcal{U}_t((\cdot)\otimes \omega^{S_t'} \otimes \omega^{S_t})] \\
      \ttran{t}(\cdot) &= \Tr_{S_t'}[\tilde{\mathcal{U}}_t((\cdot)\otimes \tilde{\omega}^{S_t'} \otimes \omega^{S_t})].
    \end{aligned}
  \end{equation}
    The dual of the involved channels are summarized in Table \ref{table:channel_dual}; a similar list of statements holds for the noisy channels. By composing the duals, the claim is proved. 
    \begin{table}[ht]
      \centering
      \begin{tabular}{|c| c|}
        \hline
        Channel & Dual \\
        \hline
        $\rho \to \Tr_{S_t'}[\rho]$ & $O \to  I_{S_t'}\otimes I$ \\
        \hline
        $\rho \to \mathcal{U}_t(\rho)$ & $O\to \mathcal{U}_t^*(O)$ \\
        \hline
        $\rho \to \rho \otimes \omega^{S_t}$ & $O \to \Tr_{S_t}[\omega^{S_t}O]$ \\
        \hline
        $\rho \to \rho \otimes \omega^{S_t'}$ & $O \to \Tr_{S_t'}[\omega^{S_t'}O]$ \\
        \hline
      \end{tabular}
      \caption{Channels and their duals.\label{table:channel_dual}}
    \end{table}
\end{proof}

The physical interpretation of the dual channel is the ``time-evolution'' in the Heisenberg picture. For a state $\rho^{BS_{[1,t]}}$ and an operator $O\in \mathcal{B}(\mathcal{H}_{BS_{[1,t+1]}})$,
\begin{equation}
\Tr[T_B^{BS_{t+1}}(\rho^{BS_{[1,t]}}) O] = \Tr[\rho^{BS_{[1,t]}} T_B^{BS_{t+1}*}(O)].
\end{equation}
Of course, $T_B^{BS_{t+1}*}$ maps an operator in $\mathcal{B}(\mathcal{H}_{BS_{[1,t+1]}})$ to operators in $\mathcal{B}(\mathcal{H}_{BS_{[1,t]}})$; see FIG.\ref{fig:Schrodinger_Heisenberg}. It is important to note that the time in the Heisenberg picture is backward. While the state evolves from $t=0$ to $t=\ell_y$, the operator evolves from $t=\ell_y$ to $t=0$.
\begin{figure}[h]
  \subfigure[Schr\"odinger picture]{
    \begin{tikzpicture}[scale=0.5]
      \node[below](t) at (4.4,0,16) {$t=5$};
    \foreach \i in {5}
        \foreach \j in {1,...,15}
                {\draw (\i,0.5,\j)--(\i,0,\j);}
    \draw (5,0.5,1)--(5,0.5,15);
    \draw (5,0,1)--(5,0,15);
    \foreach \i in {1,...,5}
             {\draw (1,0,3*\i)--(5,0,3*\i);
               \draw[dashed] (5,0,3*\i)--(7,0,3*\i);
             }
    \foreach \i in {1,...,5}
             {\draw (\i,0,3)--(\i,0,15);}
    \foreach \i in {6,7}
             {\draw[dashed] (\i,0,3)--(\i,0,15);}
    \foreach \i in {5}
        \foreach \j in {1,...,15}
                {\node[circle,draw=black,fill=white,scale=0.4](\i) at (\i, 0, \j) {};}
    \foreach \i in {1,...,7}
        \foreach \j in {1,...,5}
                {\node[circle,draw=black,fill=black,scale=0.4](\i) at (\i,0, 3*\j) {};}
    \foreach \j in {1,...,15}
             {\node[circle,draw=gray,fill=gray,scale=0.4](\j) at (5,0.5,\j) {};}
    \draw[->] (7.5,0,7.5)--(8.5,0,7.5);

             \begin{scope}[xshift=8cm]
               
      \node[below](t) at (4.4,0,16) {$t=6$};
    \foreach \i in {6}
        \foreach \j in {1,...,15}
                {\draw (\i,0.5,\j)--(\i,0,\j);}
    \draw (6,0.5,1)--(6,0.5,15);
    \draw (6,0,1)--(6,0,15);
    \foreach \i in {1,...,5}
             {\draw (1,0,3*\i)--(6,0,3*\i);
               \draw[dashed] (6,0,3*\i)--(7,0,3*\i);
             }
    \foreach \i in {1,...,6}
             {\draw (\i,0,3)--(\i,0,15);}
    \foreach \i in {7}
             {\draw[dashed] (\i,0,3)--(\i,0,15);}
    \foreach \i in {6}
        \foreach \j in {1,...,15}
                {\node[circle,draw=black,fill=white,scale=0.4](\i) at (\i, 0, \j) {};}
    \foreach \i in {1,...,7}
        \foreach \j in {1,...,5}
                {\node[circle,draw=black,fill=black,scale=0.4](\i) at (\i,0, 3*\j) {};}
    \foreach \j in {1,...,15}
             {\node[circle,draw=gray,fill=gray,scale=0.4](\j) at (6,0.5,\j) {};}
             \end{scope}
  \end{tikzpicture}}
  
  \subfigure[Heisenberg picture]{
    \begin{tikzpicture}[scale=0.5]
      
      \node[below](t) at (4.4,0,16) {$t=6$};
    \foreach \i in {6}
        \foreach \j in {1,...,15}
                {\draw (\i,0.5,\j)--(\i,0,\j);}
    \draw (6,0.5,1)--(6,0.5,15);
    \draw (6,0,1)--(6,0,15);
    \foreach \i in {1,...,5}
             {\draw (1,0,3*\i)--(6,0,3*\i);
               \draw[dashed] (6,0,3*\i)--(7,0,3*\i);
             }
    \foreach \i in {1,...,6}
             {\draw (\i,0,3)--(\i,0,15);}
    \foreach \i in {7}
             {\draw[dashed] (\i,0,3)--(\i,0,15);}
    \foreach \i in {6}
        \foreach \j in {1,...,15}
                {\node[circle,draw=black,fill=white,scale=0.4](\i) at (\i, 0, \j) {};}
    \foreach \i in {1,...,7}
        \foreach \j in {1,...,5}
                {\node[circle,draw=black,fill=black,scale=0.4](\i) at (\i,0, 3*\j) {};}
    \foreach \j in {1,...,15}
             {\node[circle,draw=gray,fill=gray,scale=0.4](\j) at (6,0.5,\j) {};}
                 \draw[->] (7.5,0,7.5)--(8.5,0,7.5);

             \begin{scope}[xshift=8cm]
               \node[below](t) at (4.4,0,16) {$t=5$};
    \foreach \i in {5}
        \foreach \j in {1,...,15}
                {\draw (\i,0.5,\j)--(\i,0,\j);}
    \draw (5,0.5,1)--(5,0.5,15);
    \draw (5,0,1)--(5,0,15);
    \foreach \i in {1,...,5}
             {\draw (1,0,3*\i)--(5,0,3*\i);
               \draw[dashed] (5,0,3*\i)--(7,0,3*\i);
             }
    \foreach \i in {1,...,5}
             {\draw (\i,0,3)--(\i,0,15);}
    \foreach \i in {6,7}
             {\draw[dashed] (\i,0,3)--(\i,0,15);}
    \foreach \i in {5}
        \foreach \j in {1,...,15}
                {\node[circle,draw=black,fill=white,scale=0.4](\i) at (\i, 0, \j) {};}
    \foreach \i in {1,...,7}
        \foreach \j in {1,...,5}
                {\node[circle,draw=black,fill=black,scale=0.4](\i) at (\i,0, 3*\j) {};}
    \foreach \j in {1,...,15}
             {\node[circle,draw=gray,fill=gray,scale=0.4](\j) at (5,0.5,\j) {};}
             \end{scope}
             
    \end{tikzpicture}
    }
  \caption{The time evolution in the Schr\"odinger(a) and the Heisenberg(b) picture. \label{fig:Schrodinger_Heisenberg}}
\end{figure}
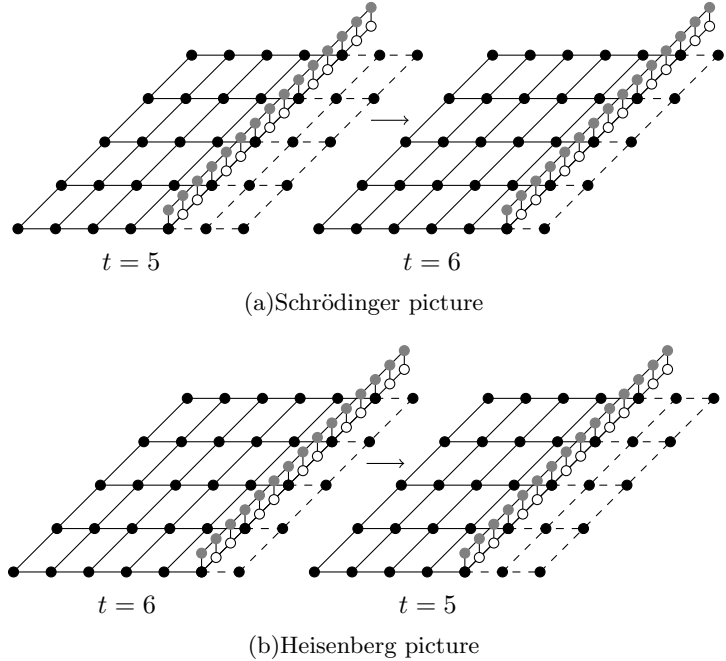

These channels preserve locality. That is, consider an operator $O$ at time $t$. If it is a local operator, then both $T_B^{BS_t*}(O)$ and $\tilde{T}_B^{BS_t*}(O)$ are local. This follows from the fact that the circuit elements that influence the operator $O$ only lies in its vicinity; see FIG.\ref{fig:locality_preserving} and Lemma \ref{lemma:locality_preserving}.
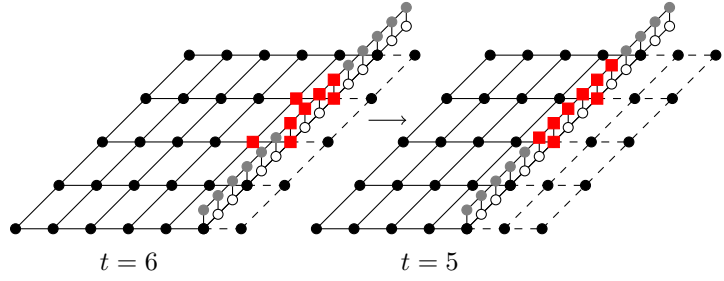
\begin{figure}[h]
\begin{tikzpicture}[scale=0.5]
  \node[below](t) at (4.4,0,16) {$t=6$};
    \foreach \i in {6}
        \foreach \j in {1,...,15}
                {\draw (\i,0.5,\j)--(\i,0,\j);}
    \draw (6,0.5,1)--(6,0.5,15);
    \draw (6,0,1)--(6,0,15);
    \foreach \i in {1,...,5}
             {\draw (1,0,3*\i)--(6,0,3*\i);
               \draw[dashed] (6,0,3*\i)--(7,0,3*\i);
             }
    \foreach \i in {1,...,6}
             {\draw (\i,0,3)--(\i,0,15);}
    \foreach \i in {7}
             {\draw[dashed] (\i,0,3)--(\i,0,15);}
    \foreach \i in {6}
        \foreach \j in {1,...,15}
                {\node[circle,draw=black,fill=white,scale=0.4](\i) at (\i, 0, \j) {};}
    \foreach \i in {1,...,7}
        \foreach \j in {1,...,5}
                {\node[circle,draw=black,fill=black,scale=0.4](\i) at (\i,0, 3*\j) {};}
    \foreach \j in {1,...,15}
             {\node[circle,draw=gray,fill=gray,scale=0.4](\j) at (6,0.5,\j) {};}
             \draw[->] (7.5,0,7.5)--(8.5,0,7.5);

             \node[minimum size=0.3cm,draw=red,fill=red,scale=0.5](1) at (5,0,6) {};
             \node[minimum size=0.3cm,draw=red,fill=red,scale=0.5](1) at (6,0,6) {};
             \node[minimum size=0.3cm,draw=red,fill=red,scale=0.5](1) at (5,0,9) {};
             \node[minimum size=0.3cm,draw=red,fill=red,scale=0.5](1) at (6,0,9) {};
             \node[minimum size=0.3cm,draw=red,fill=red,scale=0.5](1) at (6,0.5,6) {};
             \node[minimum size=0.3cm,draw=red,fill=red,scale=0.5](1) at (6,0.5,7) {};
             \node[minimum size=0.3cm,draw=red,fill=red,scale=0.5](1) at (6,0.5,8) {};
             \node[minimum size=0.3cm,draw=red,fill=red,scale=0.5](1) at (6,0.5,9) {};

             \begin{scope}[xshift=8cm]
               \node[below](t) at (4.4,0,16) {$t=5$};
    \foreach \i in {5}
        \foreach \j in {1,...,15}
                {\draw (\i,0.5,\j)--(\i,0,\j);}
    \draw (5,0.5,1)--(5,0.5,15);
    \draw (5,0,1)--(5,0,15);OB
    \foreach \i in {1,...,5}
             {\draw (1,0,3*\i)--(5,0,3*\i);
               \draw[dashed] (5,0,3*\i)--(7,0,3*\i);
             }
    \foreach \i in {1,...,5}
             {\draw (\i,0,3)--(\i,0,15);}
    \foreach \i in {6,7}
             {\draw[dashed] (\i,0,3)--(\i,0,15);}
    \foreach \i in {5}
        \foreach \j in {1,...,15}
                {\node[circle,draw=black,fill=white,scale=0.4](\i) at (\i, 0, \j) {};}
    \foreach \i in {1,...,7}
        \foreach \j in {1,...,5}
                {\node[circle,draw=black,fill=black,scale=0.4](\i) at (\i,0, 3*\j) {};}
    \foreach \j in {1,...,15}
             {\node[circle,draw=gray,fill=gray,scale=0.4](\j) at (5,0.5,\j) {};}
             
             \node[minimum size=0.3cm,draw=red,fill=red,scale=0.5](1) at (5,0,6) {};
             \node[minimum size=0.3cm,draw=red,fill=red,scale=0.5](1) at (5,0,9) {};
             \node[minimum size=0.3cm,draw=red,fill=red,scale=0.5](1) at (5,0.5,6) {};
             \node[minimum size=0.3cm,draw=red,fill=red,scale=0.5](1) at (5,0.5,7) {};
             \node[minimum size=0.3cm,draw=red,fill=red,scale=0.5](1) at (5,0.5,8) {};
             \node[minimum size=0.3cm,draw=red,fill=red,scale=0.5](1) at (5,0.5,9) {};
             \node[minimum size=0.3cm,draw=red,fill=red,scale=0.5](1) at (5,0.5,5) {};
             \node[minimum size=0.3cm,draw=red,fill=red,scale=0.5](1) at (5,0.5,10) {};

             \end{scope}
\end{tikzpicture}
\caption{In the Heisenberg picture, a local operator is mapped into another local operator. This figure describes an example of $D=1$. The red squares represent the support of the operator. \label{fig:locality_preserving}}
\end{figure}

\begin{lem}\label{lemma:locality_preserving}
  Let $O$ be an operator supported on a ball of radius $r$. Then $T_B^{BS_t*}(O)$ and $\tilde{T}_B^{BS_t*}(O)$ are supported on a ball of radius $r+D$.
\end{lem}
\begin{proof}
  By Lemma \ref{lemma:dual},
  \begin{equation}
    \begin{aligned}
      T_B^{BS_t*}(O) = \Tr_{S_tS_t'}[\omega^{S_tS_t'} \mathcal{U}_t^{(1)*}\circ \cdots \circ \mathcal{U}_t^{(D)*}(I_{S_t'} \otimes O)].
    \end{aligned}
  \end{equation}
  The support of $ \mathcal{U}_t^{(1)*}\circ \cdots \circ \mathcal{U}_t^{(D)*}(I_{S_t'} \otimes O)$ is contained in a set of qubits that are distance $D$ or less away from the support of $O$; this is because applying a channel $\mathcal{U}_t^{(i)*}$, $i\in [1,D]$, to an operator $O$ expands the support at most to a set of qubits that are distance $1$ or less away from the support of $O$. Since $O$ was assumed to be supported in a ball of radius $r$, $ \mathcal{U}_t^{(1)*}\circ \cdots \circ \mathcal{U}_t^{(D)*}(I_{S_t'} \otimes O)$ is supported in a ball of radius $r+D$. Since the support is nonincreasing under the partial trace operation, the claim is proved for $T_B^{BS_t*}(O)$. The same logic applies to $\tilde{T}_B^{BS_t*}(O)$.
\end{proof}

Furthermore, the action of $T_B^{BS_t*}$ and $\tilde{T}_B^{BS_t*}$ on local operators deviate at most by $\mathcal{O}(\epsilon)$.
\begin{lem}
\label{lemma:local_approximate_equivalence}
  Let $O$  be an operator supported on a ball of radius $r$. Then
  \begin{equation}
    \|T_B^{BS_t*}(O) - \tilde{T}_B^{BS_t*}(O) \| \leq \mathcal{O}(\epsilon \|O \|(r+D)^2).
  \end{equation}
\end{lem}
\begin{proof}
  By Lemma \ref{lemma:locality_preserving}, both $T_B^{BS_t*}(O)$ and $\tilde{T}_B^{BS_t*}(O)$ are supported on a ball of radius $r+D$. In particular, these operators can be expressed as a composition of channels localized in the vicinity of $O$. To see why, recall that
  \begin{equation}
    \begin{aligned}
      T_B^{BS_t*}(O) &= \Tr_{S_tS_t'}[\omega^{S_tS_t'} \mathcal{U}_t^{(1)*}\circ \cdots \circ \mathcal{U}_t^{(D)*}(I_{S_t'} \otimes O)].
    \end{aligned}
  \end{equation}
  Because each of the $\mathcal{U}_t^{(i)*}$ for $i\in [1,D]$ consists of nearest-neighbor gates, their action on $I_{S_t'} \otimes O$ can be replaced only in terms of the action of the channels $\mathcal{U}_t^{(i,j)}$ that influence $O$. Let us denote this channel, which is localized in the vicinity of $O$, as
  \begin{equation}
    \mathcal{U}_{t,\text{loc}}^* = \underbrace{\mathcal{U}_{t}^{(1,i_1)*}\circ \cdots \mathcal{U}_t^{(1,i_{n_1})*} \circ \cdots \circ \mathcal{U}_t^{(D,i_1)*} \circ \cdots \circ \mathcal{U}_t^{(D,i_{n_D})*}}_{n},
  \end{equation}
  where $n\leq \mathcal{O}(rD+D^2)$ is the total number of circuit elements that can influence $O$.

  The bound on $n$ is derived as follows. At each $i\in [1,D]$, the support of the operator lies on a ball of radius $r+i$. Since the circuit is only applied to $BS_tS_t'$, these circuit elements act nontrivially only on a region of size bounded by $2(r+i)$. The number of nearest-neighbor two-qubit gates that can act nontrivially on this region is $\mathcal{O}(r+i)$. Summing over $i$ from $1$ to $D$, $n\leq\mathcal{O}(rD+D^2)$.

Now, let us define 
  \begin{equation}
    \tilde{\mathcal{U}}_{t,\text{loc}}^* = \tilde{\mathcal{U}}_{t}^{(1,i_1)*}\circ \cdots \tilde{\mathcal{U}}_t^{(1,i_{n_1})*} \circ \cdots \circ \tilde{\mathcal{U}}_t^{(D,i_1)*} \circ \cdots \circ \tilde{\mathcal{U}}_t^{(D,i_{n_D})*}.
  \end{equation}
  Then
  \begin{equation}
    \begin{aligned}
      \|T_B^{BS_t*}(O) - \tilde{T}_{B}^{BS_t*}(O) \| &= |\Tr_{S_tS_t'}[\omega^{S_tS_t'} \mathcal{U}_{t,\text{loc}}(I_{S_t'}^* \otimes O) - \tilde{\omega}^{S_t}\otimes \tilde{\omega}^{S_t'} \tilde{\mathcal{U}}_{t,\text{loc}}^*(I_{S_t'} \otimes O)]| \\
      &\leq \delta_1 + \delta_2,
    \end{aligned}   
  \end{equation}
  where
  \begin{equation}
    \begin{aligned}
    &\delta_1= |\Tr_{S_tS_t'}[(\omega^{S_tS_t'} - \tilde{\omega}^{S_tS_t'}) \mathcal{U}_{t,\text{loc}}^*(I_{S_t'} \otimes O)]|  \\
    &\delta_2=|\Tr_{S_tS_t'}[\omega^{S_tS_t'}(\mathcal{U}_{t,\text{loc}}^*(I_{S_t'}\otimes O)  - \tilde{\mathcal{U}}_{t,\text{loc}}^*(I_{S_t'}\otimes O)) ]|.
    \end{aligned}
  \end{equation}
  By Lemma \ref{lemma:locality_preserving}, the support of $ \mathcal{U}_{t,\text{loc}}(I_{S_t'} \otimes O)]|$ is supported on a ball of radius $r+D$. By applying Lemma \ref{lemma:state_systemsink}, $\delta_1\leq \mathcal{O}(\epsilon\|O\|(r+D)^2)$. For the second term,
    \begin{equation}
      \begin{aligned}
        \delta_2 &\leq\|O \| \|\mathcal{U}_{t,\text{loc}}^* - \tilde{\mathcal{U}}_{t,\text{loc}}^* \|_{\diamond} \\
        &\leq \|O \|\|\mathcal{U}_{t,\text{loc}}^* - \tilde{\mathcal{U}}_{t}^{(1,i_1)*}\circ \mathcal{U}_{t,\text{loc}}^{*'} \|_{\diamond} + \|O \|\|\tilde{\mathcal{U}}_{t}^{(1,i_1)*}\circ \mathcal{U}_{t,\text{loc}}^{*'} - \tilde{\mathcal{U}}_{t,\text{loc}}^* \|_{\diamond} \\
        &\leq \|\mathcal{U}_{t}^{(1,i_1)*} - \tilde{\mathcal{U}}_t^{(1,i_1)*} \|_{\diamond} + \|\mathcal{U}_{t,\text{loc}}^{*'} - \tilde{\mathcal{U}}_{t,\text{loc}}^{*'} \|_{\diamond}, \\
        &\leq \epsilon \|O \|  + \|O \| \|\mathcal{U}_{t,\text{loc}}^{*'} - \tilde{\mathcal{U}}_{t,\text{loc}}^{*'} \|_{\diamond},
        \end{aligned}
    \end{equation}
    where
    \begin{equation}
      \begin{aligned}
    \mathcal{U}_{t,\text{loc}}^{*'} &= \underbrace{\mathcal{U}_{t}^{(1,i_2)*}\circ \cdots \mathcal{U}_t^{(1,i_{n_1})*} \circ \cdots \circ \mathcal{U}_t^{(D,i_1)*} \circ \cdots \circ \mathcal{U}_t^{(D,i_{n_D})*}}_{n-1}\\
    \tilde{\mathcal{U}}_{t,\text{loc}}^{*'} &= \underbrace{\tilde{\mathcal{U}}_{t}^{(1,i_2)*}\circ \cdots \tilde{\mathcal{U}}_t^{(1,i_{n_1})*} \circ \cdots \circ \tilde{\mathcal{U}}_t^{(D,i_1)*} \circ \cdots \circ \tilde{\mathcal{U}}_t^{(D,i_{n_D})*}}_{n-1}.
      \end{aligned}
      \end{equation}
      By iterating this bound $n=\mathcal{O}(rD+D^2)$ times, $\delta_2\leq \mathcal{O}(\epsilon\|O \|(rD+D^2))$. Thus, $\delta_1+\delta_2 \leq \mathcal{O}(\epsilon \|O \|(r+D)^2)$.
\end{proof}
\subsection{Operators\label{section:observables}}
Here we estimate the difference between $O$ and $\tilde{O}$.
\begin{lem}
  \label{lemma:observables}
  Let $O$ be a Pauli operator with $|\Supp{O}|=n$.
  \begin{equation}
    \| O - \tilde{O} \| \leq n\epsilon.
  \end{equation}
\end{lem}
\begin{proof}
   Without loss of generality, let $O= \otimes_{i=1}^{n} \sigma^{a_i}$, $\tilde{O}=\otimes_{i=1}^{n} \tilde{\sigma}^{a_i}$ and define $O_j= \otimes_{i=j}^n \sigma^{a_i}$, $\tilde{O}_j = \otimes_{i=j}^n \tilde{\sigma}^{a_i}$. Then
\begin{equation}
  \begin{aligned}
    \| O - \tilde{O} \| &\leq \|\sigma^{a_1}\otimes (O_2 - \tilde{O}_2)\| + \|(\sigma^{a_1} - \tilde{\sigma}^{a_1})\otimes \tilde{O}_2 \| \\
    &\leq \|O_2- \tilde{O}_2\| + \epsilon \\
    &\leq n\epsilon.
  \end{aligned}
\end{equation}
  \end{proof}

\section{Reduction to the bath dynamics \label{section:Heisenberg}}
In this Section, we relate the difference between the noiseless and the noisy expectation value to the difference between two time-dependent expectation values. Specifically, for a normalized local operator $O$,
\begin{equation}
  |\Tr[ \rho O ] - \Tr[\tilde{\rho} \tilde{O}] |\leq \mathcal{O}(\epsilon) +  |\Tr[\rho^B(t) O'] - \Tr[\tilde{\rho}^B(t) O'],
\end{equation}
where $O'$ is a normalized operator on the bath. The support of $O'$ and $t$ are determined by the support and the location of $O$. The dynamics of $\rho^B(t)$ is defined in terms of $T_B^{BS_t*}$. Similarly, the dynamics of  $\tilde{\rho}^B(t)$ is defined in terms of  $\tilde{T}_B^{BS_t*}$. 

This reduction is useful because the dynamics can be viewed as a discrete-time analogue of a certain dissipative dynamics. Its stability to extensive perturbation has been studied\cite{Cubitt2013,Lucia2014}, and we can use the techniques developed therein to prove the main result; this is the content of Section \ref{section:local_uniform_contraction}.

There is an intuitive explanation. Without loss of generality, consider an observable $O$ supported on a ball of bounded radius. Since the dual of $T_B^{BS_t}$(as well as the dual of $\tilde{T}_B^{BS_t}$) maps operators supported on $BS_t$ to operators on $B$, $O$ is eventually mapped into an operator supported only on the bath. This process takes time proportional to the radius of the support of $O$; see FIG.\ref{fig:reduction}. Because the radius was assumed to be bounded, the process takes a bounded amount of time. Within this time, the support of $O$ can be only expanded by a bounded amount, because the channels are applied by applying a bounded number of finite-depth quantum circuits. This means that the number of circuit elements that can influence $O$ within this time is bounded, independent of the system size. The effect of noise on these circuit elements can be conservatively bounded by summing over each of their contributions, which is in the order of $\mathcal{O}(\epsilon)$; the observable undergoing a noisy time evolution, up to this point, is close to the observable undergoing a noiseless time evolution up to an error of $\mathcal{O}(\epsilon)$. Furthermore, at this point, the observable is supported on a localized region of the bath. The closeness of the noisy and noiseless expectation value can thus be analyzed by studying the behavior of local observables supported on the bath.
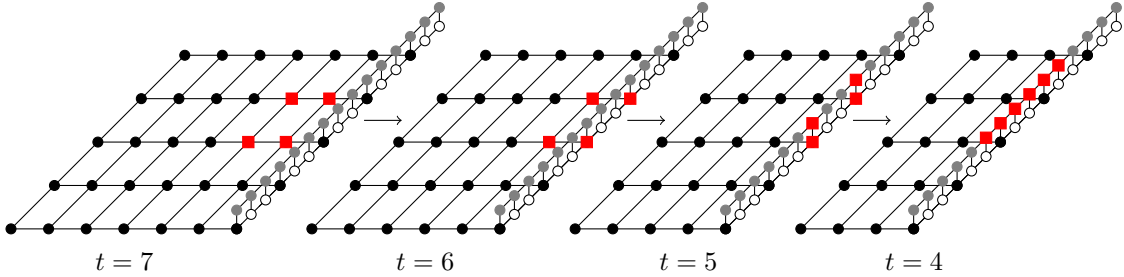
\begin{figure}[h]
\begin{tikzpicture}[scale=0.5]
  \node[below](t) at (4.4,0,16) {$t=7$};
    \foreach \i in {7}
        \foreach \j in {1,...,15}
                {\draw (\i,0.5,\j)--(\i,0,\j);}
    \draw (7,0.5,1)--(7,0.5,15);
    \draw (7,0,1)--(7,0,15);
    \foreach \i in {1,...,5}
             {\draw (1,0,3*\i)--(7,0,3*\i);
               \draw[dashed] (7,0,3*\i)--(7,0,3*\i);
             }
    \foreach \i in {1,...,6}
             {\draw (\i,0,3)--(\i,0,15);}
    \foreach \i in {7}
             {\draw[dashed] (\i,0,3)--(\i,0,15);}
    \foreach \i in {7}
        \foreach \j in {1,...,15}
                {\node[circle,draw=black,fill=white,scale=0.4](\i) at (\i, 0, \j) {};}
    \foreach \i in {1,...,7}
        \foreach \j in {1,...,5}
                {\node[circle,draw=black,fill=black,scale=0.4](\i) at (\i,0, 3*\j) {};}
    \foreach \j in {1,...,15}
             {\node[circle,draw=gray,fill=gray,scale=0.4](\j) at (7,0.5,\j) {};}
             \draw[->] (7.5,0,7.5)--(8.5,0,7.5);

             \node[minimum size=0.3cm,draw=red,fill=red,scale=0.5](1) at (5,0,6) {};
             \node[minimum size=0.3cm,draw=red,fill=red,scale=0.5](1) at (6,0,6) {};
             \node[minimum size=0.3cm,draw=red,fill=red,scale=0.5](1) at (5,0,9) {};
             \node[minimum size=0.3cm,draw=red,fill=red,scale=0.5](1) at (6,0,9) {};

             \begin{scope}[xshift=8cm]
               \node[below](t) at (4.4,0,16) {$t=6$};
    \foreach \i in {6}
        \foreach \j in {1,...,15}
                {\draw (\i,0.5,\j)--(\i,0,\j);}
    \draw (6,0.5,1)--(6,0.5,15);
    \draw (6,0,1)--(6,0,15);
    \foreach \i in {1,...,5}
             {\draw (1,0,3*\i)--(6,0,3*\i);
             }
    \foreach \i in {1,...,6}
             {\draw (\i,0,3)--(\i,0,15);}
    \foreach \i in {6}
        \foreach \j in {1,...,15}
                {\node[circle,draw=black,fill=white,scale=0.4](\i) at (\i, 0, \j) {};}
    \foreach \i in {1,...,6}
        \foreach \j in {1,...,5}
                {\node[circle,draw=black,fill=black,scale=0.4](\i) at (\i,0, 3*\j) {};}
    \foreach \j in {1,...,15}
             {\node[circle,draw=gray,fill=gray,scale=0.4](\j) at (6,0.5,\j) {};}
             
             \node[minimum size=0.3cm,draw=red,fill=red,scale=0.5](1) at (5,0,6) {};
             \node[minimum size=0.3cm,draw=red,fill=red,scale=0.5](1) at (6,0,6) {};
             \node[minimum size=0.3cm,draw=red,fill=red,scale=0.5](1) at (5,0,9) {};
             \node[minimum size=0.3cm,draw=red,fill=red,scale=0.5](1) at (6,0,9) {};
 \draw[->] (6.5,0,7.5)--(7.5,0,7.5);

             \end{scope}
             
             \begin{scope}[xshift=15cm]
               \node[below](t) at (4.4,0,16) {$t=5$};
    \foreach \i in {5}
        \foreach \j in {1,...,15}
                {\draw (\i,0.5,\j)--(\i,0,\j);}
    \draw (5,0.5,1)--(5,0.5,15);
    \draw (5,0,1)--(5,0,15);
    \foreach \i in {1,...,5}
             {\draw (1,0,3*\i)--(5,0,3*\i);
             }
    \foreach \i in {1,...,5}
             {\draw (\i,0,3)--(\i,0,15);}
    \foreach \i in {5}
        \foreach \j in {1,...,15}
                {\node[circle,draw=black,fill=white,scale=0.4](\i) at (\i, 0, \j) {};}
    \foreach \i in {1,...,5}
        \foreach \j in {1,...,5}
                {\node[circle,draw=black,fill=black,scale=0.4](\i) at (\i,0, 3*\j) {};}
    \foreach \j in {1,...,15}
             {\node[circle,draw=gray,fill=gray,scale=0.4](\j) at (5,0.5,\j) {};}
             
             \node[minimum size=0.3cm,draw=red,fill=red,scale=0.5](1) at (5,0,6) {};
             \node[minimum size=0.3cm,draw=red,fill=red,scale=0.5](1) at (5,0,9) {};
             \node[minimum size=0.3cm,draw=red,fill=red,scale=0.5](1) at (5,0.5,6) {};
             \node[minimum size=0.3cm,draw=red,fill=red,scale=0.5](1) at (5,0.5,9) {};
                 \draw[->] (5.5,0,7.5)--(6.5,0,7.5);
             \end{scope}
             \begin{scope}[xshift=21cm]
               \node[below](t) at (4.4,0,16) {$t=4$};
    \foreach \i in {4}
        \foreach \j in {1,...,15}
                {\draw (\i,0.5,\j)--(\i,0,\j);}
    \draw (4,0.5,1)--(4,0.5,15);
    \draw (4,0,1)--(4,0,15);
    \foreach \i in {1,...,5}
             {\draw (1,0,3*\i)--(4,0,3*\i);
             }
    \foreach \i in {1,...,4}
             {\draw (\i,0,3)--(\i,0,15);}
    \foreach \i in {4}
        \foreach \j in {1,...,15}
                {\node[circle,draw=black,fill=white,scale=0.4](\i) at (\i, 0, \j) {};}
    \foreach \i in {1,...,4}
        \foreach \j in {1,...,5}
                {\node[circle,draw=black,fill=black,scale=0.4](\i) at (\i,0, 3*\j) {};}
    \foreach \j in {1,...,15}
             {\node[circle,draw=gray,fill=gray,scale=0.4](\j) at (4,0.5,\j) {};}
             
             \node[minimum size=0.3cm,draw=red,fill=red,scale=0.5](1) at (4,0.5,6) {};
             \node[minimum size=0.3cm,draw=red,fill=red,scale=0.5](1) at (4,0.5,9) {};
             
             \node[minimum size=0.3cm,draw=red,fill=red,scale=0.5](1) at (4,0.5,7) {};
             \node[minimum size=0.3cm,draw=red,fill=red,scale=0.5](1) at (4,0.5,8) {};
             \node[minimum size=0.3cm,draw=red,fill=red,scale=0.5](1) at (4,0.5,10) {};
             \node[minimum size=0.3cm,draw=red,fill=red,scale=0.5](1) at (4,0.5,5) {};
             \end{scope}
\end{tikzpicture}
\caption{In the Heisenberg picture, a local operator supported on the system is mapped into a local operator supported on a system-bath composite, and then to a local operator in bath. The red squares represent the support of an operator. Here $D=1$.\label{fig:reduction}}
\end{figure}

Once the operator is mapped into an operator supported only on the bath, its time evolution is determined by the following channel:
\begin{equation}
  \begin{aligned}
    T_{t}(\cdot) &= \Tr_{S_t}[T_B^{BS_t}(\cdot)] \\
    \tilde{T}_t(\cdot) &= \Tr_{S_t}[\tilde{T}_B^{BS_t}(\cdot)].
  \end{aligned}
\end{equation}
Also,
\begin{equation}
  \label{eq:bath_dynamics}
  \begin{aligned}
    T_{[t,t']}(\cdot) &= \Tr_{S_{[t,t']}}[T_B^{BS_{[t,t']}}(\cdot)]  \\
    \tilde{T}_{[t,t']}(\cdot) &= \Tr_{S_{[t,t']}}[\tilde{T}_B^{BS_{[t,t']}}(\cdot)].
  \end{aligned}
\end{equation}
Below, the main result of this Section is stated and proved.

\begin{lem}
  For a Pauli operator $O$ supported on a ball of radius $r$ from coordinate $t=t_0+1$ to $t=t_0+2r$,
  \begin{equation}
    |\Tr[\rho O - \tilde{\rho}\tilde O]| \leq \mathcal{O}(\epsilon r^3D^2\|O \|) + \delta,
    \end{equation}
where 
\begin{equation}
  \delta =\|O \| \sup_{\substack{\|O' \|\leq 1}}|\Tr[T_{[1,t_0]}^*(\rho^B) O' - \tilde{T}_{[1,t_0]}^*(\tilde{\rho}^B) O']|,
\end{equation}
where $O'\in \mathcal{B}(\mathcal{H}_B)$ is an operator supported on a ball of radius $\mathcal{O}(Dr)$. 
\label{lemma:Heisenberg}
\end{lem}
\begin{proof}
  \begin{equation}
    \begin{aligned}
      |\Tr[\rho O - \tilde{\rho} \tilde{O}]| &\leq |\Tr[(\rho - \tilde{\rho}) O]| + |\Tr[\tilde{\rho}(O-\tilde{O})]|.
    \end{aligned}
  \end{equation}
  The first term is bounded as
  \begin{equation}
    \begin{aligned}
      |\Tr[(\rho -\tilde{\rho})O]| &= |\Tr[(\trant{1}{\ell_y}(\rho^B) - \ttrant{1}{\ell_y}(\tilde{\rho}^B))O]| \\
      &=|\Tr[\rho^B\trantd{1}{\ell_y}(O) - \tilde{\rho}^B\ttrantd{1}{\ell_y}(O)]| \\
      &=|\Tr[\rho^B\trantd{1}{t_0+2r}(O) - \tilde{\rho}^B\ttrantd{1}{t_0+2r}(O)]| \\
      &=|\Tr[\rho^B\trantd{1}{t_0}\circ \trantd{t_0+1}{t_0+2r}(O) - \tilde{\rho}^B\ttrantd{1}{t_0}\circ \ttrantd{t_0+1}{t_0+2r}(O)]| \\
      &=|\Tr[\rho^BT_{[1,t_0]}^*\circ \trantd{t_0+1}{t_0+2r}(O) - \tilde{\rho}^B\tilde{T}_{[1,t_0]}^*\circ \ttrantd{t_0+1}{t_0+2r}(O)]| \\
     &\leq |\Tr[\rho^BT_{[1,t_0]}^*(O_{t_0+1} - O_{t_0+1}') ]| + |\Tr[\rho^BT_{[1,t_0]}^*(O_{t_0+1}')] - \Tr[\tilde{\rho}^B \tilde{T}_{[1,t_0]}^*(O_{t_0+1}')]|,
    \end{aligned}
  \end{equation}
  where
  \begin{equation}
    \begin{aligned}
    O_{t} &= \trantd{t}{\ell_y}(O)\\
    O_{t}' &= \ttrantd{t}{\ell_y}(O).
    \end{aligned}
  \end{equation}
  Then,
  \begin{equation}
    \begin{aligned}
      \|O_{t_0+1} - O_{t_0+1}'\| &\leq \|O_{t_0+1}- \tilde{T}_B^{BS_{t_0+1}*}(O_{t_0+2}) \|  + \|\tilde{T}_B^{BS_{t_0+1}*}(O_{t_0+2}) - O_{t_0+1}' \|  \\
      &\leq \mathcal{O}(\epsilon\|O \|r^2D^2) + \|\tilde{T}_B^{BS_{t_0+1}*}(O_{t_0+2} - O_{t_0+2}') \|\\
      &\leq \mathcal{O}(\epsilon\|O \|r^2D^2) + \|O_{t_0+2} - O_{t_0+2}' \|,
    \end{aligned}
  \end{equation}
  where in the second line we used Lemma \ref{lemma:locality_preserving} and \ref{lemma:local_approximate_equivalence}. Specifically, the action of the noiseless and noisy channel on $O_{t_0+2}$ is bounded by $\mathcal{O}(\epsilon \|O_{t_0+2}\| (R+D)^2)$, where $O_{t_0+2}$ is supported on a ball of radius $R$. Since $\|O_{t_0+2}\| \leq \|O \|$ and $R$ is bounded by $r(1+D)$, the dominant term is $\mathcal{O}(\epsilon \|O \|r^2D^2)$. The third line follows from the fact that quantum channels are norm-nonincreasing. By iterating this bound,
  \begin{equation}
    \begin{aligned}
      |\Tr[(\rho - \tilde{\rho})O ]| &\leq \mathcal{O}(\epsilon \|O \|r^3D^2) + |\Tr[\rho^BT_{[1,t_0]}^*(O_{t_0+1}')] - \Tr[\tilde{\rho}^B \tilde{T}_{[1,t_0]}^*(O_{t_0+1}')]|,
    \end{aligned}
  \end{equation}
  where $O_{t_0+1}$ is supported on a ball of radius $\mathcal{O}(Dr)$. 
  
    Lastly, by Lemma \ref{lemma:observables}
    \begin{equation}
      \|O - \tilde{O}\| \leq \mathcal{O}(\epsilon r^2\|O \| ).
    \end{equation}
    Combining these bounds, the claim is proved.
\end{proof}

\section{Local rapid mixing\label{section:local_uniform_contraction}}
In Section \ref{section:Heisenberg}, the difference between the noiseless and noisy local expectation value was bounded in terms of two time-dependent expectation values;  up to an $\mathcal{O}(\epsilon)$ correction, the effect of noise on local expectation value is bounded by  $\sup_{\|O\|\leq 1} |\Tr[\rho^BT_{[1,t_0]}^*(O)] - \Tr[\tilde{\rho}^B\tilde{T}_{[1,t_0]}^*(O)]|$, where $O$ is a local observable supported on the bath. In this Section, this term is bounded by an object which only depends on the properties of the noiseless time evolution($T_{[t,t']}$). This bound holds for any circuits, but becomes nontrivial only for a certain family. Circuits belonging to such family are noise-resilient.

This bound is expressed in terms of the \emph{contraction of $T_{[t,t']}$ at a certain length scale}.
\begin{defi}
  A contraction of $T_{[t,t']}$ at length scale $\ell$ is
  \begin{equation}
    \eta^{\ell}(T_{[t,t']}) = \sup_{\substack{\| O \|\leq 1 \\ \Supp{O} =A}}\| T_{[t,t']}^*(O) - \Phi \circ T_{[t,t']}^*(O) \|,
  \end{equation}
for any $A$ which is a ball of radius $\ell$, where $\Phi$ is the completely depolarizing channel. We say $T_{[t,t']}$ is $(\Delta, \ell_0)$-locally rapidly mixing if 
  \begin{equation}
    \eta^{\ell}(T_{[t,t']}) \leq c\ell^{\alpha}e^{-\gamma t} + \Delta
  \end{equation}
  for all $\ell\leq \ell_0$, where $c,\alpha,\gamma>0$ are numerical constants.
\end{defi}
We show that
\begin{thm}
  \label{thm:local_rapid_mixing}
  If $T_{[t,t']}$ is  $(\Delta, \ell_0)-$locally rapidly mixing for some $c>0$ with $D=\mathcal{O}(1)$, then for any observable $O$ supported on a ball of radius $r=\mathcal{O}(1)$, $r<\ell_0$,
  \begin{equation}
    |\Tr[T_{[1,t_0]}(\rho^B) O - \tilde{T}_{[1,t_0]}(\tilde{\rho}^B)O]| \leq \mathcal{O}(\epsilon \log^2(1/\epsilon) + \ell_y\Delta)\|O \|.
    \end{equation}
  \end{thm}

  Let us first sketch the proof. Note that
  \begin{equation}
    |\Tr[T_{[1,t]}(\rho^B) O - \tilde{T}_{[1,t_0]}(\tilde{\rho}^B)O] | \leq \delta_1 + \delta_2,
  \end{equation}
  where
  \begin{equation}
    \begin{aligned}
      \delta_1 &= |\Tr[(\rho^B - \tilde{\rho}^B) T_{[1,t]}^*(O)]|\\
      \delta_2 &= |\Tr[\tilde{\rho}^B(T^{*}_{[1,t]}(O) - \tilde{T}_{[1,t]}^*(O))]|
    \end{aligned} 
  \end{equation}
  by the triangle inequality. At small $t$, $\delta_1$ is bounded by $\mathcal{O}(\epsilon)$ because (i) $T_{[1,t]}(O)$ is a local observable and (ii) local expectation values of $\rho^B$ and $\tilde{\rho^B}$ deviate at most by $\mathcal{O}(\epsilon)$(cf. Lemma \ref{lemma:state_bath}). Also, $\delta_2$ is bounded by $\mathcal{O}(\epsilon)$ because the action of $T^{*}_{[1,t]}$ and $\tilde{T}_{[1,t]}^*$ on local observables differ at most by $\mathcal{O}(\epsilon)$(cf. Lemma \ref{lemma:local_approximate_equivalence}). At large $t$, for $\delta_1$, $T_{[1,t]}^*(O)$ can be approximated by $\Phi \circ T_{[1,t]}^*(O)$ up to an error $\Delta$. This operator is the (rescaled) identity operator, so its expectation value is independent of the state. Therefore, $\delta_1$ is bounded by $\Delta$. For $\delta_2$, note the following identity:
  \begin{equation}
    T_{[1,t]}^* - \tilde{T}_{[1,t]}^* = \sum_{n=1}^t\tilde{T}_{n-1}^* \circ (T_n^* - \tilde{T}_n^*)\circ T_{[n+1,t]}^*,
  \end{equation}
  where $\tilde{T}^*_0 = \mathcal{I}$ and $T_{[t+1,t]}^*=\mathcal{I}$. At small $n$, $T_{[n+1,t]}^*(O)$ can be approximated by $\Phi\circ T_{[n+1,t]}^*(O)$ up to an error $\Delta$. Since $T_n^*$ and $\tilde{T}_n^*$ are unital, their action on $\Phi\circ T_{[n+1,t]}^*(O)$ is identical. What remains is an error term which is bounded by $\Delta$. At $n$ close to $t$, $T_{[n+1,t]}$ is a local operator, and by Lemma \ref{lemma:local_approximate_equivalence} the action of $T_n^*$ and $\tilde{T}_n^*$ on this operator differs at most by $\mathcal{O}(\epsilon)$.

  Of course, this logic applies only to small and large $t$. What remains to be shown is a bound that interpolates between these two limits. Below, we derive such bounds for $\delta_1$ and $\delta_2$. Then Theorem \ref{thm:local_rapid_mixing} follows from these bounds.

  \begin{lem}
    Let $T_{[t,t']}$ is $(\Delta, \ell_0)-$locally rapidly mixing with $D=\mathcal{O}(1)$. For an operator $O$ supported on a ball of radius $r=\mathcal{O}(1)$ where $r<\ell_0$,  
    \begin{equation}
      \delta_1 \leq \mathcal{O}((\epsilon \log(1/\epsilon) + \Delta) \|O \|)
    \end{equation}
  \end{lem}

\begin{proof}
  Note that $T_{[1,t_0]}^*(O)$ is supported on a ball of radius $r+tD$.(cf. Lemma \ref{lemma:locality_preserving}) By Lemma \ref{lemma:state_bath},
  \begin{equation}
    \delta_1 \leq \mathcal{O}(\epsilon \|O \|(r+tD)).
  \end{equation}
  This bound becomes rather lousy at large $t$. In the large $t$ regime, the following bound works better.
  \begin{equation}
    \begin{aligned}
      |\Tr[(\rho^B-\tilde{\rho}^B)T_{[1,t]}^*(O)]| &= |\Tr[(\rho^B-\tilde{\rho}^B) (T_{[1,t]}^*(O) -     \Phi \circ T_{[1,t]}^*(O))]| \\
      &\leq \|\rho^B - \tilde{\rho}^B\|_1 \|T_{[1,t]}^*(O) - \Phi\circ T_{[1,t]}^*(O) \| \\
      &\leq 2c\|O \|r^{\alpha} e^{-\gamma t} + 2\|O\|\Delta.
    \end{aligned}
    \end{equation}
    Both bounds hold for all $t$, so we can take the minimum of the two. The function $e^{-\gamma t}$ decreases monotonically and $tD$ increases monotonically. Therefore, the minimum of the two bounds is bounded from above by $\epsilon(r+t^*D)$, where $\epsilon t^*D = c'r^{\alpha}e^{-\gamma t^*}$, where $c'$ is some numerical constant.
    \begin{equation}
      \begin{aligned}
        t^{*} &=\frac{1}{\gamma}\log \frac{c'r^{\alpha}}{\epsilon  t^* D},
        \end{aligned}
      \end{equation}
      which is bounded from above by $\frac{1}{\gamma} \log \frac{c'r^{\alpha}}{\epsilon D}$ for $t^*\geq 1$ and $1$ if $t^*\leq 1$. 
\end{proof}

\begin{lem}
  Let $T_{[t,t']}$ is $(\Delta, \ell_0)$-locally rapidly mixing with $D=\mathcal{O}(1)$. For an operator $O$ supported on a ball of radius $r=\mathcal{O}(1)$ where $r<\ell_0$,
  \begin{equation}
    \delta_2 \leq \mathcal{O}((\epsilon \log^2(1/\epsilon) + \ell_y \Delta)\|O\|)
  \end{equation}
\end{lem}
\begin{proof}
  \begin{equation}
    \delta_2 \leq \sum_{n=1}^t \|T_n^*(T_{[n+1,t]}^*(O)) - \tilde{T}_n^*(T_{[n+1,t]}^*(O)) \|.
  \end{equation}
  The summand can be bounded in two ways. First,
   \begin{equation}
    \|T_n^*(T_{[n+1,t]}^*(O)) - \tilde{T}_n^*(T_{[n+1,t]}^*(O)) \| \leq \mathcal{O}(\epsilon \|O \|(t-n-1)).
  \end{equation}
  This is because $T_{[n+1,t]}^*(O)$ is supported on a ball of radius $r+D(n+1-t)$ by Lemma \ref{lemma:locality_preserving}. By applying Lemma \ref{lemma:local_approximate_equivalence} and restricting the action of the channel to the bath, we arrive at this bound.

  Second, for 
  \begin{equation}
    \begin{aligned}
      \|T_n^*(T_{[n+1,t]}^*(O)) - \tilde{T}_n^*(T_{[n+1,t]}^*(O)) \| &\leq 2\|T_{[n+1,t]}^*(O) - \Phi \circ T_{[n+1,t]}^*(O) \|\\
      &\leq 2c\|O \|r^{\alpha} e^{-\gamma (t-n-1)} + 2\|O\|\Delta,
    \end{aligned}
  \end{equation}
  because the summand is equal to $ \|T_n^*(T_{[n+1,t]}^*(O) - \Phi\circ T_{[n+1,t]}^*(O)) - \tilde{T}^*_n(T_{[n+1,t]}^*(O) - \Phi \circ T_{[n+1,t]}^*(O)) \|$.

  The sum can be bounded as
  \begin{equation}
    \begin{aligned}
      \delta_2 &\leq \|O \| \sum_{n=1}^t \min(\epsilon c' n, ce^{-\gamma n} + 2\Delta) \\
      &\leq \|O \|( \sum_{n=1}^{t*} \epsilon c' n +\sum_{n=t^*}^{t} (ce^{-\gamma n} + 2\Delta)) \\
      &\leq \|O \| (\mathcal{O}(\epsilon t^{*2} + e^{-\gamma t^*}) + \sum_{n=t^*}^t 2\Delta )\\
      &\leq \|O \|(\mathcal{O}(\epsilon t^{*2} + e^{-\gamma t^*}) + 2\Delta \ell_y),
      \end{aligned}
    \end{equation}
    where $c'$ is some constant  $t^*$ is the largest $t$ such that $\epsilon t^* \leq c'e^{-\gamma t^*}$. Note that
    \begin{equation}
      t^* = \frac{1}{\gamma} \log \frac{c'}{\epsilon t^*}.
    \end{equation}
    If $t^*\geq 1$, then $t^*=\mathcal{O}(\log(1/\epsilon))$. Also,
    \begin{equation}
      \epsilon (t^*+1) \geq c' e^{-\gamma (t^*+1)}.
    \end{equation}
    Therefore, $e^{-\gamma t^*} = \mathcal{O}(\epsilon \log(1/\epsilon))$. Combining these bounds,
    \begin{equation}
      \delta_2 \leq \|O \|(\mathcal{O}(\epsilon \log^2(1/\epsilon) + \ell_y \Delta))
      \end{equation}
\end{proof}

Now that we have Theorem \ref{thm:local_rapid_mixing} and Lemma \ref{lemma:Heisenberg}, we can combine these bounds together to conclude that
\begin{thm}
  \label{thm:local_rapid_mixing_stability}
  If $T_{[t,t']}$ is $(\Delta,\ell_0)$-locally rapidly mixing for some $c>0$ with $D= \mathcal{O}(1)$, then for any observable $O$ supported on a ball of radius $r=\mathcal{O}(1)$, $r<\ell_0$,
  \begin{equation}
    |\Tr[\rho O - \tilde{\rho} \tilde{O}]|\leq \mathcal{O}(\epsilon \log^{2}(1/\epsilon) + \ell_y \Delta) \|O \|.
  \end{equation}
\end{thm}

\subsection{Comments on the bound}
In Theorem \ref{thm:local_rapid_mixing_stability}, one may wonder why we need the $\ell_y \Delta$ term, and whether this term is controlled in general. To answer the first question, this term is absolutely necessary to describe topologically ordered systems that are away from the fixed-point wavefunction. If the system size is finite, the expectation value of local observables over different ground states differ, albeit only by an exponentially small amount. The term $\Delta$ upper bounds this difference, and as such, cannot be made to be exactly zero. Nevertheless, we do not expect this to pose a problem because $\Delta$ is expected to decay exponentially in the system size, effectively suppressing the $\ell_y$ contribution.

\section{Examples \label{section:examples}}
In this Section, we provide two nontrivial examples that satisfy the local rapid mixing condition. We study a circuit that prepares the ground state of the toric code\cite{Kitaev2003} with open boundary condition\cite{Bravyi1998} and a circuit that prepares a trivial state.

\subsection{Surface Code}
Let us consider toric code with an open boundary condition, also known as the surface code. For the notational convenience, we decided to place the qubits on the vertices as opposed to the edges. The local terms of the Hamiltonian consist of four types: four-body plaquette terms of $X/Z$ type and two-body boundary terms of $X/Z$ type. As an example, a $5\times 5$ surface code is described in FIG.\ref{fig:RSC}. Below, we show that the surface code on a $\ell \times \ell$ lattice is $(0,\ell-1)-$locally rapidly mixing. The argument is based on the stabilizer formalism\cite{Gottesman1997}.
\begin{figure}[h]
\begin{tikzpicture}[scale=0.9]
\foreach \x in {0,...,4}
{
  \foreach \y in {0,...,4}
  {
    \node[circle, draw=black, fill=black, scale=0.5]() at (1.0 * \x, 1.0 * \y) {};
  }
}
\foreach \x in {0,...,4}
{
  \draw (\x, 0)--(\x,4);
}
\foreach \y in {0,...,4}
{
  \draw (0,\y)--(4,\y);
}
\foreach \x in {0,2}
{
  \foreach \y in {0,2}
  {
    \node[] at (\x+0.5,\y+0.5) {X};
  }
}
\foreach \x in {1,3}
{
  \foreach \y in {1,3}
  {
    \node[] at (\x+0.5,\y+0.5) {X};
  }
}
\foreach \x in {0,2}
{
  \foreach \y in {1,3}
  {
    \node[] at (\x+0.5,\y+0.5) {Z};
  }
}
\foreach \x in {1,3}
{
  \foreach \y in {0,2}
  {
    \node[] at (\x+0.5,\y+0.5) {Z};
  }
}
\foreach \xx in {1,3}
{
  \draw[domain=180:360] plot ({\xx +0.5 + 0.5 * cos(\x)}, {0.5 * sin(\x)});
  \draw[domain=0:180] plot ({\xx -0.5 + 0.5 * cos(\x)}, {0.5 * sin(\x)+4});
  \node[above] at (\xx +0.5, -0.5 ) {X};
  \node[below] at (\xx -0.5, 4.5) {X};
}
\foreach \yy in {1,3}
{
  \draw[domain=-90:90] plot ({4 + 0.5 * cos(\x)}, {\yy + 0.5 + 0.5 * sin(\x)});
  \draw[domain=90:270] plot ({ 0.5 * cos(\x)}, {\yy - 0.5 + 0.5 * sin(\x)});
  \node[left] at (4.5, \yy + 0.5) {Z};
  \node[right] at (-0.5, \yy - 0.5) {Z};
}

\end{tikzpicture}
\caption{Toric code with an open boundary condition, rotated by $\frac{\pi}{4}$, on a $5\times 5$ lattice. The qubits are located at the vertices(black dots). In the bulk, the Hamiltonian consists of plaquette terms of $X$ and $Z$ type. Two-body terms, i.e., $XX$ and $ZZ$, are placed on the physical boundary.  \label{fig:RSC}}
\end{figure}
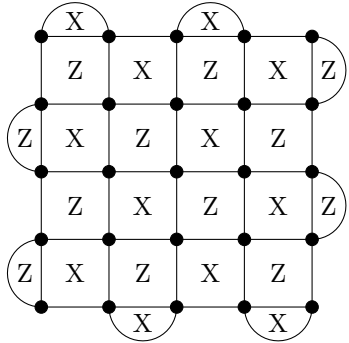

Any product state can be gradually transformed to the ground state, row-by-row. Intuitively, the process can be thought as measuring each of the local terms in the Hamiltonian that are supported on two contiguous rows, and applying an appropriate set of single-qubit gates so that all the plaquette and boundary terms are satisfied.\cite{Dennis2001} The circuits are depicted in FIG.\ref{fig:prepcircuit_toric_code_Z} and \ref{fig:prepcircuit_toric_code_X}. The boundary terms can be treated in a similar manner, by removing qubit $1$ and $2$, as well as the CNOTs that act on these qubits. By setting the $Z$-type stabilizers to $+1$ in parallel, and then setting $X$-type stabilizers to $+1$ in parallel, one can extend the code on $n$ rows to a larger code on $(n+1)$ rows.
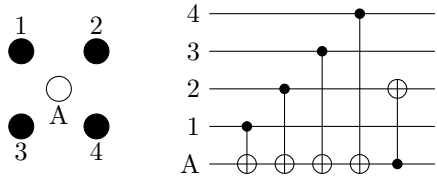
\begin{figure}[h]
\begin{tikzpicture}
\foreach \x in {0,1}
{
  \foreach \y in {0,1}
  {
    \node[circle, draw=black, fill=black]() at (1.0 * \x, 1.0 * \y) {};
  }
}
\node[circle, draw=black, fill=white]() at (0.5,0.5) {};
\node[above] at (0,1.1) {1};
\node[above] at (1,1.1) {2};
\node[below] at (0,-0.1) {3};
\node[below] at (1,-0.1) {4};
\node[below] at (0.5,0.4){A};
\begin{scope}[xshift=2.5cm,yshift=-0.5cm]
\foreach \y in {0,...,4}
{
 \draw (0,\y*0.5)--(3,\y*0.5);
}
\node[left] at (0,0) {A};
\node[left] at (0,0.5) {1};
\node[left] at (0,1.0) {2};
\node[left] at (0,1.5) {3};
\node[left] at (0,2.0) {4};
\foreach \x in {1,...,4}
{
  \draw( \x*0.5, -0.125)--(\x*0.5, \x*0.5);
  \draw (\x*0.5, 0 ) circle (0.125cm);
  \filldraw (\x*0.5, \x*0.5) circle (0.0625cm);
}
\draw (2.5,0)--(2.5,1.125);
\draw (2.5,1.0) circle (0.125cm);
\filldraw (2.5, 0) circle (0.0625cm);
\end{scope}
\end{tikzpicture}
\caption{Subroutine for initializing the stabilizer of $Z$ type. The circuit effectively measure the $Z$-type stabilizer and apply the correction so that the stabilizer yields $+1$, no matter what the initial state is. The ancilla($A$) is initialized to $\ket{0}$.\label{fig:prepcircuit_toric_code_Z}}
\end{figure}
\begin{figure}[h]
\begin{tikzpicture}
\foreach \x in {0,1}
{
  \foreach \y in {0,1}
  {
    \node[circle, draw=black, fill=black]() at (1.0 * \x, 1.0 * \y) {};
  }
}
\node[circle, draw=black, fill=white]() at (0.5,0.5) {};
\node[above] at (0,1.1) {1};
\node[above] at (1,1.1) {2};
\node[below] at (0,-0.1) {3};
\node[below] at (1,-0.1) {4};
\node[below] at (0.5,0.4){A};
\begin{scope}[xshift=2.5cm,yshift=-0.5cm]
\foreach \y in {0,...,4}
{
 \draw (0,\y*0.5)--(3,\y*0.5);
}
\node[left] at (0,0) {A};
\node[left] at (0,0.5) {1};
\node[left] at (0,1.0) {2};
\node[left] at (0,1.5) {3};
\node[left] at (0,2.0) {4};
\foreach \x in {1,...,4}
{
  \draw( \x*0.5, 0)--(\x*0.5, \x*0.5+0.125);
  \draw (\x*0.5, \x*0.5 ) circle (0.125cm);
  \filldraw (\x*0.5, 0) circle (0.0625cm);
}
\draw (2.5,-0.125)--(2.5,1.0);
\filldraw (2.5, 1.0) circle (0.0625cm);
\draw (2.5, 0 ) circle (0.125cm);
\end{scope}
\end{tikzpicture}
\caption{Subroutine for initializing the stabilizer of $X$ type. The circuit effectively measure the $X$-type stabilizer and apply the correction so that the stabilizer yields $+1$, no matter what the initial state is. The ancilla($A$) is initialized to $\ket{+}=(\ket{0}+\ket{1})/\sqrt{2}$.\label{fig:prepcircuit_toric_code_X}}
\end{figure}
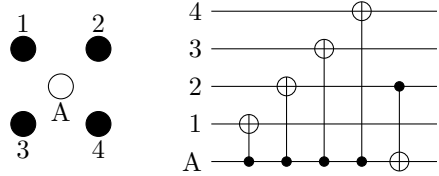

Note that the circuits that set the same type of stabilizers on different plaquettes do not overlap with each other. This implies that they can be run in parallel. Therefore, each row can be prepared by a bounded-depth circuit, by first setting the $X$-type stabilizers to the $+1$ eigenstate and then setting the $Z$-type stabilizers to the $+1$ eigenstate.

How does this circuit fit into our framework? At time $t\leq \ell_y-2$, the system qubits are the physical qubits supported on the $(t)$-th row, the bath qubits are the physical qubits supported on the $(t+1)$-th row, and the sink qubits are the physical qubits supported on the $(t+2)$-th row and the ancillas that appear in FIG.\ref{fig:prepcircuit_toric_code_Z} and \ref{fig:prepcircuit_toric_code_X}. At the beginning of each time step, the bath qubit is swapped with the (initialized) system qubit, and then the circuits in FIG.\ref{fig:prepcircuit_toric_code_Z} and \ref{fig:prepcircuit_toric_code_X} are applied in order to set the stabilizers supported on the three contiguous rows to be $+1$. Then the sink qubits are initialized. For $t=\ell_y-1$ and  $\ell_y$, stabilizers supported on two contiguous rows are set to $+1$. The rest of the protocol, i.e., initialization of the system qubit and the sink qubits, is identical.

The state of the bath becomes a reduced density matrix of one of the ground states over a single row of physical qubits, \emph{independent of the state of the bath qubit at the previous time step.} To see why, note that applying the circuits in FIG.~\ref{fig:prepcircuit_toric_code_Z} and \ref{fig:prepcircuit_toric_code_X} to three contiguous rows of qubits set the stabilizer generators to be $+1$ independent of the state over the physical qubits. This is because (i) each of the circuits are constructed in such a way that, if the stabilizer is not $+1$, then a correction operation(the last CNOT gate) is performed to flip its value to $+1$ and (ii) these circuits do not alter the value of the stabilizer generators in the previous step. This means that the stabilizers surrounding the middle row middle row are all set to $+1$, independent of the previous history. Since this row corresponds to the bath qubits, the state of the bath becomes a reduced density matrix of one of the ground states over a single row of physical qubits. 

This implies that the circuit is $(0,\ell_x-1)$-locally rapidly mixing. Except for the logical operator, which is the string of Pauli-$Z$ operator along the row, every Pauli operator supported on a single row anticommutes with the stabilizer generators surrounding that row. Therefore, for any Pauli operator $P$ supported on a single row, with the exception of the logical operator and the identity operator,
\begin{equation}
T_n^*(P) = 0,
\end{equation}
which subsequently implies that $\eta^{\ell}(T_{[t,t']})=0$ for $\ell<\ell_x$.

\subsection{Trivial state}
What happens if the underlying state is a trivial product state? Similar to the preparation of the surface code, at time $t$, the system qubits are the physical qubits supported on the $(t)$-th row and the bath qubits are the physical qubits supported on the $(t+1)$-th row. After the system qubits are measured, they are initialized and then swapped with the bath qubits. Then a local unitary transformation is applied in the bath qubits to prepare a state of the next row. In the bath, this amounts to preparing a fixed state at each $t$, say $\rho^B_t$. The dual of this map is $O \to \Phi(\rho^B_t O)$, where $\Phi$ is the completely depolarizing channel. Therefore, again the local rapid mixing condition is satisfied. 

It is important to note that the local rapid mixing condition is not satisfied if we do not swap the bath and the system qubits. Then the action of the circuit on the bath is unitary, and as such, the condition cannot be satisfied unless the unitary is trivial. The state preparation remains to be robust to noise for the trivial state, but this shows that the local rapid mixing condition is really a statement about the circuit, not the underlying state. It will be interesting if there is a condition that is formulated in terms of the underlying state and not the circuit, such that ensures the stability of the local expectation values.

\section{Discussion\label{section:Discussion}}
We introduced a notion of noise resilience for large-depth quantum circuits. Even without error correction, noise can affect the expectation values of local observables at most by an amount comparable to the noise rate, independent of the system size. We provided concrete examples of such circuits.

The fact that such a circuit exists is interesting in its own right, because it is counter to the intuition that the effect of noise on large-depth quantum circuits are uncontrolled. What is even more remarkable is that certain circuits that prepare ground states of realistic physical systems possess this property. An immediate corollary is that a low-energy state of such systems can be prepared on a noisy quantum simulator. Furthermore, this state can reproduce the expectation values of all the local observables up to a fixed precision.

The main takeaway message should be that certain large-depth quantum circuits are resilient to noise. By identifying a class of such circuits, we can explore the possibilities offered  by a noisy quantum computer that lie beyond the paradigm of short-depth quantum circuits.

We end with some open problems to pursue.
\begin{enumerate}
\item  What kind of physical states can be prepared by a circuit that obeys a nontrivial local rapid mixing condition? We have shown that the trivial state and the surface code belong to this family, but a larger class of models\cite{Levin2005} may be amenable to a similar analysis.

\item Can we noise-resiliently prepare ground states that are away from the fixed point? If two states are adiabatically connected to each other, then they can be mapped into each other by the so called quasi-adiabatic continuation\cite{Hastings2005}. While there is a sense in which this preserves locality, there is no known method to approximate such map by a finite-depth local unitary operation. It will be interesting to understand how our conclusion changes.

\item Can we arrive at the same conclusion just from the property of the state? It is important to remember that noise-resilience in our context is a statement about a \emph{circuit that prepares the ground state}, and not the ground state itself. Decay of correlation in the state is bounded in terms of the properties of the operator $T_{[t,t']}$, but the reverse direction generally does not work. It will be nice to know if one can arrive at the same conclusion assuming that correlation decays exponentially in the ground state.

\item Given a Hamiltonian, suppose we found a low-energy configuration whose energy per site is $\mathcal{O}(\epsilon)$ away from the ground state. Does this state reliably approximate the ground state local expectation values? Since the number of states at that energy is exponentially large in the system size, the answer is not completely clear. It will be interesting to formulate a condition under which this is true. The eigenstate thermalization hypothesis\cite{Deutsch1991,Srednicki1994}, or alternatively, a recently proposed notion of canonical universality\cite{Dymarsky2017} may prove useful. 
\end{enumerate}
\begin{acknowledgments}
I thank Kristan Temme and Jay Gambetta for helpful discussions.
\end{acknowledgments}

\bibliography{bib}

\begin{thebibliography}{33}%
\makeatletter
\providecommand \@ifxundefined [1]{%
 \@ifx{#1\undefined}
}%
\providecommand \@ifnum [1]{%
 \ifnum #1\expandafter \@firstoftwo
 \else \expandafter \@secondoftwo
 \fi
}%
\providecommand \@ifx [1]{%
 \ifx #1\expandafter \@firstoftwo
 \else \expandafter \@secondoftwo
 \fi
}%
\providecommand \natexlab [1]{#1}%
\providecommand \enquote  [1]{``#1''}%
\providecommand \bibnamefont  [1]{#1}%
\providecommand \bibfnamefont [1]{#1}%
\providecommand \citenamefont [1]{#1}%
\providecommand \href@noop [0]{\@secondoftwo}%
\providecommand \href [0]{\begingroup \@sanitize@url \@href}%
\providecommand \@href[1]{\@@startlink{#1}\@@href}%
\providecommand \@@href[1]{\endgroup#1\@@endlink}%
\providecommand \@sanitize@url [0]{\catcode `\\12\catcode `\$12\catcode
  `\&12\catcode `\#12\catcode `\^12\catcode `\_12\catcode `\%12\relax}%
\providecommand \@@startlink[1]{}%
\providecommand \@@endlink[0]{}%
\providecommand \url  [0]{\begingroup\@sanitize@url \@url }%
\providecommand \@url [1]{\endgroup\@href {#1}{\urlprefix }}%
\providecommand \urlprefix  [0]{URL }%
\providecommand \Eprint [0]{\href }%
\providecommand \doibase [0]{http://dx.doi.org/}%
\providecommand \selectlanguage [0]{\@gobble}%
\providecommand \bibinfo  [0]{\@secondoftwo}%
\providecommand \bibfield  [0]{\@secondoftwo}%
\providecommand \translation [1]{[#1]}%
\providecommand \BibitemOpen [0]{}%
\providecommand \bibitemStop [0]{}%
\providecommand \bibitemNoStop [0]{.\EOS\space}%
\providecommand \EOS [0]{\spacefactor3000\relax}%
\providecommand \BibitemShut  [1]{\csname bibitem#1\endcsname}%
\let\auto@bib@innerbib\@empty
\bibitem [{\citenamefont {Kim}(2017)}]{Kim2017a}%
  \BibitemOpen
  \bibfield  {author} {\bibinfo {author} {\bibfnamefont {Isaac~H.}\
  \bibnamefont {Kim}},\ }\bibfield  {title} {\enquote {\bibinfo {title}
  {Holographic quantum simulation},}\ }\href@noop {} {\  (\bibinfo {year}
  {2017})},\ \Eprint {http://arxiv.org/abs/1702.02093v1} {1702.02093v1}
  \BibitemShut {NoStop}%
\bibitem [{\citenamefont {Troyer}\ and\ \citenamefont
  {Wiese}(2005)}]{Troyer2005}%
  \BibitemOpen
  \bibfield  {author} {\bibinfo {author} {\bibfnamefont {Matthias}\
  \bibnamefont {Troyer}}\ and\ \bibinfo {author} {\bibfnamefont {Uwe-Jens}\
  \bibnamefont {Wiese}},\ }\bibfield  {title} {\enquote {\bibinfo {title}
  {Computational complexity and fundamental limitations to fermionic quantum
  monte carlo simulations},}\ }\href {\doibase 10.1103/PhysRevLett.94.170201}
  {\bibfield  {journal} {\bibinfo  {journal} {Phys. Rev. Lett.}\ }\textbf
  {\bibinfo {volume} {94}},\ \bibinfo {pages} {170201} (\bibinfo {year}
  {2005})}\BibitemShut {NoStop}%
\bibitem [{\citenamefont {White}(1992)}]{White1992}%
  \BibitemOpen
  \bibfield  {author} {\bibinfo {author} {\bibfnamefont {Steven~R.}\
  \bibnamefont {White}},\ }\bibfield  {title} {\enquote {\bibinfo {title}
  {Density matrix formulation for quantum renormalization groups},}\ }\href
  {\doibase 10.1103/PhysRevLett.69.2863} {\bibfield  {journal} {\bibinfo
  {journal} {Phys. Rev. Lett.}\ }\textbf {\bibinfo {volume} {69}},\ \bibinfo
  {pages} {2863--2866} (\bibinfo {year} {1992})}\BibitemShut {NoStop}%
\bibitem [{\citenamefont {Verstraete}\ and\ \citenamefont
  {Cirac}(2004)}]{Verstraete2004a}%
  \BibitemOpen
  \bibfield  {author} {\bibinfo {author} {\bibfnamefont {F.}~\bibnamefont
  {Verstraete}}\ and\ \bibinfo {author} {\bibfnamefont {J.~I.}\ \bibnamefont
  {Cirac}},\ }\bibfield  {title} {\enquote {\bibinfo {title} {Renormalization
  algorithms for quantum-many body systems in two and higher dimensions},}\
  }\href@noop {} {\  (\bibinfo {year} {2004})},\ \Eprint
  {http://arxiv.org/abs/cond-mat/0407066v1} {cond-mat/0407066v1} \BibitemShut
  {NoStop}%
\bibitem [{\citenamefont {Schuch}\ \emph {et~al.}(2007)\citenamefont {Schuch},
  \citenamefont {Wolf}, \citenamefont {Verstraete},\ and\ \citenamefont
  {Cirac}}]{Schuch2007}%
  \BibitemOpen
  \bibfield  {author} {\bibinfo {author} {\bibfnamefont {Norbert}\ \bibnamefont
  {Schuch}}, \bibinfo {author} {\bibfnamefont {Michael~M.}\ \bibnamefont
  {Wolf}}, \bibinfo {author} {\bibfnamefont {Frank}\ \bibnamefont
  {Verstraete}}, \ and\ \bibinfo {author} {\bibfnamefont {J.~Ignacio}\
  \bibnamefont {Cirac}},\ }\bibfield  {title} {\enquote {\bibinfo {title}
  {Computational complexity of projected entangled pair states},}\ }\href
  {\doibase 10.1103/PhysRevLett.98.140506} {\bibfield  {journal} {\bibinfo
  {journal} {Phys. Rev. Lett.}\ }\textbf {\bibinfo {volume} {98}},\ \bibinfo
  {pages} {140506} (\bibinfo {year} {2007})}\BibitemShut {NoStop}%
\bibitem [{\citenamefont {Lubasch}\ \emph {et~al.}(2014)\citenamefont
  {Lubasch}, \citenamefont {Cirac},\ and\ \citenamefont
  {Bañuls}}]{Lubasch2014}%
  \BibitemOpen
  \bibfield  {author} {\bibinfo {author} {\bibfnamefont {Michael}\ \bibnamefont
  {Lubasch}}, \bibinfo {author} {\bibfnamefont {J~Ignacio}\ \bibnamefont
  {Cirac}}, \ and\ \bibinfo {author} {\bibfnamefont {Mari-Carmen}\ \bibnamefont
  {Bañuls}},\ }\bibfield  {title} {\enquote {\bibinfo {title} {Unifying
  projected entangled pair state contractions},}\ }\href
  {http://stacks.iop.org/1367-2630/16/i=3/a=033014} {\bibfield  {journal}
  {\bibinfo  {journal} {New J. Phys.}\ }\textbf {\bibinfo {volume} {16}},\
  \bibinfo {pages} {033014} (\bibinfo {year} {2014})}\BibitemShut {NoStop}%
\bibitem [{\citenamefont {Vidal}(2008)}]{Vidal2008}%
  \BibitemOpen
  \bibfield  {author} {\bibinfo {author} {\bibfnamefont {G.}~\bibnamefont
  {Vidal}},\ }\bibfield  {title} {\enquote {\bibinfo {title} {Class of quantum
  many-body states that can be efficiently simulated},}\ }\href {\doibase
  10.1103/PhysRevLett.101.110501} {\bibfield  {journal} {\bibinfo  {journal}
  {Phys. Rev. Lett.}\ }\textbf {\bibinfo {volume} {101}},\ \bibinfo {pages}
  {110501} (\bibinfo {year} {2008})}\BibitemShut {NoStop}%
\bibitem [{\citenamefont {Evenbly}\ and\ \citenamefont
  {Vidal}(2009)}]{Evenbly2009}%
  \BibitemOpen
  \bibfield  {author} {\bibinfo {author} {\bibfnamefont {G.}~\bibnamefont
  {Evenbly}}\ and\ \bibinfo {author} {\bibfnamefont {G.}~\bibnamefont
  {Vidal}},\ }\bibfield  {title} {\enquote {\bibinfo {title} {Algorithms for
  entanglement renormalization},}\ }\href {\doibase 10.1103/PhysRevB.79.144108}
  {\bibfield  {journal} {\bibinfo  {journal} {Phys. Rev. B}\ }\textbf {\bibinfo
  {volume} {79}},\ \bibinfo {pages} {144108} (\bibinfo {year}
  {2009})}\BibitemShut {NoStop}%
\bibitem [{Note1()}]{Note1}%
  \BibitemOpen
  \bibinfo {note} {However, there might be a more efficient method that can
  leverage the power of modern graphics processing unit.}\BibitemShut {Stop}%
\bibitem [{\citenamefont {Shor}(1996)}]{Shor1996}%
  \BibitemOpen
  \bibfield  {author} {\bibinfo {author} {\bibfnamefont {Peter~W.}\
  \bibnamefont {Shor}},\ }\bibfield  {title} {\enquote {\bibinfo {title}
  {Fault-tolerant quantum computation},}\ }\href@noop {} {\bibfield  {journal}
  {\bibinfo  {journal} {37th Symposium on Foundations of Computing, IEEE
  Computer Society Press}\ ,\ \bibinfo {pages} {56--65}} (\bibinfo {year}
  {1996})},\ \Eprint {http://arxiv.org/abs/quant-ph/9605011v2}
  {quant-ph/9605011v2} \BibitemShut {NoStop}%
\bibitem [{\citenamefont {Lloyd}(1996)}]{Lloyd1996}%
  \BibitemOpen
  \bibfield  {author} {\bibinfo {author} {\bibfnamefont {Seth}\ \bibnamefont
  {Lloyd}},\ }\bibfield  {title} {\enquote {\bibinfo {title} {Universal quantum
  simulators},}\ }\href {\doibase 10.1126/science.273.5278.1073} {\bibfield
  {journal} {\bibinfo  {journal} {Science}\ }\textbf {\bibinfo {volume}
  {273}},\ \bibinfo {pages} {1073--1078} (\bibinfo {year} {1996})}\BibitemShut
  {NoStop}%
\bibitem [{\citenamefont {Benhelm}\ \emph {et~al.}(2008)\citenamefont
  {Benhelm}, \citenamefont {Kirchmair}, \citenamefont {Roos},\ and\
  \citenamefont {Blatt}}]{Benhelm2008}%
  \BibitemOpen
  \bibfield  {author} {\bibinfo {author} {\bibfnamefont {J.}~\bibnamefont
  {Benhelm}}, \bibinfo {author} {\bibfnamefont {G.}~\bibnamefont {Kirchmair}},
  \bibinfo {author} {\bibfnamefont {C.~F.}\ \bibnamefont {Roos}}, \ and\
  \bibinfo {author} {\bibfnamefont {R.}~\bibnamefont {Blatt}},\ }\bibfield
  {title} {\enquote {\bibinfo {title} {Towards fault-tolerant quantum computing
  with trapped ions},}\ }\href {\doibase 10.1038/nphys961} {\bibfield
  {journal} {\bibinfo  {journal} {Nature Physics}\ }\textbf {\bibinfo {volume}
  {4}},\ \bibinfo {pages} {463} (\bibinfo {year} {2008})},\ \Eprint
  {http://arxiv.org/abs/0803.2798v1} {0803.2798v1} \BibitemShut {NoStop}%
\bibitem [{\citenamefont {Barends}\ \emph {et~al.}(2014)\citenamefont
  {Barends}, \citenamefont {Kelly}, \citenamefont {Megrant}, \citenamefont
  {Veitia}, \citenamefont {Sank}, \citenamefont {Jeffrey}, \citenamefont
  {White}, \citenamefont {Mutus}, \citenamefont {Fowler}, \citenamefont
  {Campbell}, \citenamefont {Chen}, \citenamefont {Chen}, \citenamefont
  {Chiaro}, \citenamefont {Dunsworth}, \citenamefont {Neill}, \citenamefont
  {O`Malley}, \citenamefont {Roushan}, \citenamefont {Vainsencher},
  \citenamefont {Wenner}, \citenamefont {Korotkov}, \citenamefont {Cleland},\
  and\ \citenamefont {Martinis}}]{Barends2014}%
  \BibitemOpen
  \bibfield  {author} {\bibinfo {author} {\bibfnamefont {R.}~\bibnamefont
  {Barends}}, \bibinfo {author} {\bibfnamefont {J.}~\bibnamefont {Kelly}},
  \bibinfo {author} {\bibfnamefont {A.}~\bibnamefont {Megrant}}, \bibinfo
  {author} {\bibfnamefont {A.}~\bibnamefont {Veitia}}, \bibinfo {author}
  {\bibfnamefont {D.}~\bibnamefont {Sank}}, \bibinfo {author} {\bibfnamefont
  {E.}~\bibnamefont {Jeffrey}}, \bibinfo {author} {\bibfnamefont {T.~C.}\
  \bibnamefont {White}}, \bibinfo {author} {\bibfnamefont {J.}~\bibnamefont
  {Mutus}}, \bibinfo {author} {\bibfnamefont {A.~G.}\ \bibnamefont {Fowler}},
  \bibinfo {author} {\bibfnamefont {B.}~\bibnamefont {Campbell}}, \bibinfo
  {author} {\bibfnamefont {Y.}~\bibnamefont {Chen}}, \bibinfo {author}
  {\bibfnamefont {Z.}~\bibnamefont {Chen}}, \bibinfo {author} {\bibfnamefont
  {B.}~\bibnamefont {Chiaro}}, \bibinfo {author} {\bibfnamefont
  {A.}~\bibnamefont {Dunsworth}}, \bibinfo {author} {\bibfnamefont
  {C.}~\bibnamefont {Neill}}, \bibinfo {author} {\bibfnamefont
  {P.}~\bibnamefont {O`Malley}}, \bibinfo {author} {\bibfnamefont
  {P.}~\bibnamefont {Roushan}}, \bibinfo {author} {\bibfnamefont
  {A.}~\bibnamefont {Vainsencher}}, \bibinfo {author} {\bibfnamefont
  {J.}~\bibnamefont {Wenner}}, \bibinfo {author} {\bibfnamefont {A.~N.}\
  \bibnamefont {Korotkov}}, \bibinfo {author} {\bibfnamefont {A.~N.}\
  \bibnamefont {Cleland}}, \ and\ \bibinfo {author} {\bibfnamefont {John~M.}\
  \bibnamefont {Martinis}},\ }\bibfield  {title} {\enquote {\bibinfo {title}
  {Logic gates at the surface code threshold: Superconducting qubits poised for
  fault-tolerant quantum computing},}\ }\href {\doibase 10.1038/nature13171}
  {\bibfield  {journal} {\bibinfo  {journal} {Nature}\ }\textbf {\bibinfo
  {volume} {508}},\ \bibinfo {pages} {500--503} (\bibinfo {year} {2014})},\
  \Eprint {http://arxiv.org/abs/1402.4848v1} {1402.4848v1} \BibitemShut
  {NoStop}%
\bibitem [{\citenamefont {Gottesman}()}]{Gottesman2016}%
  \BibitemOpen
  \bibfield  {author} {\bibinfo {author} {\bibfnamefont {Daniel}\ \bibnamefont
  {Gottesman}},\ }\bibfield  {title} {\enquote {\bibinfo {title} {Quantum fault
  tolerance in small experiments},}\ }\href@noop {} {\ }\Eprint
  {http://arxiv.org/abs/1610.03507v2} {1610.03507v2} \BibitemShut {NoStop}%
\bibitem [{\citenamefont {Fowler}\ \emph {et~al.}(2012)\citenamefont {Fowler},
  \citenamefont {Mariantoni}, \citenamefont {Martinis},\ and\ \citenamefont
  {Cleland}}]{Fowler2012}%
  \BibitemOpen
  \bibfield  {author} {\bibinfo {author} {\bibfnamefont {Austin~G.}\
  \bibnamefont {Fowler}}, \bibinfo {author} {\bibfnamefont {Matteo}\
  \bibnamefont {Mariantoni}}, \bibinfo {author} {\bibfnamefont {John~M.}\
  \bibnamefont {Martinis}}, \ and\ \bibinfo {author} {\bibfnamefont
  {Andrew~N.}\ \bibnamefont {Cleland}},\ }\bibfield  {title} {\enquote
  {\bibinfo {title} {Surface codes: Towards practical large-scale quantum
  computation},}\ }\href {\doibase 10.1103/PhysRevA.86.032324} {\bibfield
  {journal} {\bibinfo  {journal} {Phys. Rev. A}\ }\textbf {\bibinfo {volume}
  {86}},\ \bibinfo {pages} {032324} (\bibinfo {year} {2012})},\ \Eprint
  {http://arxiv.org/abs/1208.0928v2} {1208.0928v2} \BibitemShut {NoStop}%
\bibitem [{\citenamefont {Bravyi}\ and\ \citenamefont
  {Kitaev}(2005)}]{Bravyi2005}%
  \BibitemOpen
  \bibfield  {author} {\bibinfo {author} {\bibfnamefont {Sergey}\ \bibnamefont
  {Bravyi}}\ and\ \bibinfo {author} {\bibfnamefont {Alexei}\ \bibnamefont
  {Kitaev}},\ }\bibfield  {title} {\enquote {\bibinfo {title} {Universal
  quantum computation with ideal clifford gates and noisy ancillas},}\ }\href
  {\doibase 10.1103/PhysRevA.71.022316} {\bibfield  {journal} {\bibinfo
  {journal} {Phys. Rev. A}\ }\textbf {\bibinfo {volume} {71}},\ \bibinfo
  {pages} {022316} (\bibinfo {year} {2005})}\BibitemShut {NoStop}%
\bibitem [{\citenamefont {Kim}\ \emph {et~al.}(2010)\citenamefont {Kim},
  \citenamefont {Chang}, \citenamefont {Korenblit}, \citenamefont {Islam},
  \citenamefont {Edwards}, \citenamefont {Freericks}, \citenamefont {Lin},
  \citenamefont {Duan},\ and\ \citenamefont {Monroe}}]{Kim2010}%
  \BibitemOpen
  \bibfield  {author} {\bibinfo {author} {\bibfnamefont {K.}~\bibnamefont
  {Kim}}, \bibinfo {author} {\bibfnamefont {M.-S.}\ \bibnamefont {Chang}},
  \bibinfo {author} {\bibfnamefont {S.}~\bibnamefont {Korenblit}}, \bibinfo
  {author} {\bibfnamefont {R.}~\bibnamefont {Islam}}, \bibinfo {author}
  {\bibfnamefont {E.~E.}\ \bibnamefont {Edwards}}, \bibinfo {author}
  {\bibfnamefont {J.~K.}\ \bibnamefont {Freericks}}, \bibinfo {author}
  {\bibfnamefont {G.-D.}\ \bibnamefont {Lin}}, \bibinfo {author} {\bibfnamefont
  {L.-M.}\ \bibnamefont {Duan}}, \ and\ \bibinfo {author} {\bibfnamefont
  {C.}~\bibnamefont {Monroe}},\ }\bibfield  {title} {\enquote {\bibinfo {title}
  {Quantum simulation of frustrated ising spins with trapped ions},}\ }\href
  {\doibase doi:10.1038/nature09071} {\bibfield  {journal} {\bibinfo  {journal}
  {Nature}\ }\textbf {\bibinfo {volume} {465}},\ \bibinfo {pages} {590--593}
  (\bibinfo {year} {2010})}\BibitemShut {NoStop}%
\bibitem [{\citenamefont {Aharonov}\ and\ \citenamefont
  {Ben-Or}(1996)}]{Aharonov1996}%
  \BibitemOpen
  \bibfield  {author} {\bibinfo {author} {\bibfnamefont {Dorit}\ \bibnamefont
  {Aharonov}}\ and\ \bibinfo {author} {\bibfnamefont {Michael}\ \bibnamefont
  {Ben-Or}},\ }\bibfield  {title} {\enquote {\bibinfo {title} {Fault tolerant
  quantum computation with constant error},}\ }\href@noop {} {\bibfield
  {journal} {\bibinfo  {journal} {Proc. 29th Ann. ACM Symp. on Theory of
  Computing}\ ,\ \bibinfo {pages} {176}} (\bibinfo {year} {1996})},\ \Eprint
  {http://arxiv.org/abs/quant-ph/9611025v2} {quant-ph/9611025v2} \BibitemShut
  {NoStop}%
\bibitem [{\citenamefont {Cubitt}\ \emph {et~al.}(2015)\citenamefont {Cubitt},
  \citenamefont {Lucia}, \citenamefont {Michalakis},\ and\ \citenamefont
  {Perez-Garcia}}]{Cubitt2013}%
  \BibitemOpen
  \bibfield  {author} {\bibinfo {author} {\bibfnamefont {Toby~S.}\ \bibnamefont
  {Cubitt}}, \bibinfo {author} {\bibfnamefont {Angelo}\ \bibnamefont {Lucia}},
  \bibinfo {author} {\bibfnamefont {Spyridon}\ \bibnamefont {Michalakis}}, \
  and\ \bibinfo {author} {\bibfnamefont {David}\ \bibnamefont {Perez-Garcia}},\
  }\bibfield  {title} {\enquote {\bibinfo {title} {Stability of local quantum
  dissipative systems},}\ }\href {\doibase 10.1007/s00220-015-2355-3}
  {\bibfield  {journal} {\bibinfo  {journal} {Comm. Math. Phys.}\ }\textbf
  {\bibinfo {volume} {337}},\ \bibinfo {pages} {1275--1315} (\bibinfo {year}
  {2015})},\ \Eprint {http://arxiv.org/abs/1303.4744v4} {1303.4744v4}
  \BibitemShut {NoStop}%
\bibitem [{\citenamefont {Lucia}\ \emph {et~al.}(2015)\citenamefont {Lucia},
  \citenamefont {Cubitt}, \citenamefont {Michalakis},\ and\ \citenamefont
  {P\'erez-Garc\`ia}}]{Lucia2014}%
  \BibitemOpen
  \bibfield  {author} {\bibinfo {author} {\bibfnamefont {Angelo}\ \bibnamefont
  {Lucia}}, \bibinfo {author} {\bibfnamefont {Toby~S.}\ \bibnamefont {Cubitt}},
  \bibinfo {author} {\bibfnamefont {Spyridon}\ \bibnamefont {Michalakis}}, \
  and\ \bibinfo {author} {\bibfnamefont {David}\ \bibnamefont
  {P\'erez-Garc\`ia}},\ }\bibfield  {title} {\enquote {\bibinfo {title} {Rapid
  mixing and stability of quantum dissipative systems},}\ }\href {\doibase
  10.1103/PhysRevA.91.040302} {\bibfield  {journal} {\bibinfo  {journal} {Phys.
  Rev. A}\ }\textbf {\bibinfo {volume} {91}},\ \bibinfo {pages} {040302}
  (\bibinfo {year} {2015})},\ \Eprint {http://arxiv.org/abs/1409.7809v3}
  {1409.7809v3} \BibitemShut {NoStop}%
\bibitem [{\citenamefont {Fannes}\ \emph {et~al.}(1992)\citenamefont {Fannes},
  \citenamefont {Nachtergaele},\ and\ \citenamefont {Werner}}]{Fannes1992}%
  \BibitemOpen
  \bibfield  {author} {\bibinfo {author} {\bibfnamefont {M.}~\bibnamefont
  {Fannes}}, \bibinfo {author} {\bibfnamefont {B.}~\bibnamefont
  {Nachtergaele}}, \ and\ \bibinfo {author} {\bibfnamefont {R.~F.}\
  \bibnamefont {Werner}},\ }\bibfield  {title} {\enquote {\bibinfo {title}
  {Finitely correlated states on quantum spin chains},}\ }\href
  {http://projecteuclid.org/euclid.cmp/1104249404} {\bibfield  {journal}
  {\bibinfo  {journal} {Comm. Math. Phys.}\ }\textbf {\bibinfo {volume}
  {144}},\ \bibinfo {pages} {443--490} (\bibinfo {year} {1992})}\BibitemShut
  {NoStop}%
\bibitem [{\citenamefont {Verstraete}\ \emph {et~al.}(2004)\citenamefont
  {Verstraete}, \citenamefont {Garc{\'{i}}a-Ripoll},\ and\ \citenamefont
  {Cirac}}]{Verstraete2004}%
  \BibitemOpen
  \bibfield  {author} {\bibinfo {author} {\bibfnamefont {F.}~\bibnamefont
  {Verstraete}}, \bibinfo {author} {\bibfnamefont {J.~J.}\ \bibnamefont
  {Garc{\'{i}}a-Ripoll}}, \ and\ \bibinfo {author} {\bibfnamefont {J.~I.}\
  \bibnamefont {Cirac}},\ }\bibfield  {title} {\enquote {\bibinfo {title}
  {{Matrix Product Density Operators: Simulation of Finite-Temperature and
  Dissipative Systems}},}\ }\href {\doibase 10.1103/PhysRevLett.93.207204}
  {\bibfield  {journal} {\bibinfo  {journal} {Phys. Rev. Lett.}\ }\textbf
  {\bibinfo {volume} {93}},\ \bibinfo {pages} {207204} (\bibinfo {year}
  {2004})}\BibitemShut {NoStop}%
\bibitem [{\citenamefont {Zwolak}\ and\ \citenamefont
  {Vidal}(2004)}]{Zwolak2004}%
  \BibitemOpen
  \bibfield  {author} {\bibinfo {author} {\bibfnamefont {Michael}\ \bibnamefont
  {Zwolak}}\ and\ \bibinfo {author} {\bibfnamefont {Guifr{\'{e}}}\ \bibnamefont
  {Vidal}},\ }\bibfield  {title} {\enquote {\bibinfo {title} {{Mixed-state
  dynamics in one-dimensional quantum lattice systems: a time-dependent
  superoperator renormalization algorithm.}}}\ }\href {\doibase
  10.1103/PhysRevLett.93.207205} {\bibfield  {journal} {\bibinfo  {journal}
  {Phys. Rev. Lett.}\ }\textbf {\bibinfo {volume} {93}},\ \bibinfo {pages}
  {207205} (\bibinfo {year} {2004})}\BibitemShut {NoStop}%
\bibitem [{\citenamefont {Kitaev}\ \emph {et~al.}(2002)\citenamefont {Kitaev},
  \citenamefont {Shen},\ and\ \citenamefont {Vyalyi}}]{Kitaev2002}%
  \BibitemOpen
  \bibfield  {author} {\bibinfo {author} {\bibfnamefont {A.~Yu.}\ \bibnamefont
  {Kitaev}}, \bibinfo {author} {\bibfnamefont {A.~H.}\ \bibnamefont {Shen}}, \
  and\ \bibinfo {author} {\bibfnamefont {M.~N.}\ \bibnamefont {Vyalyi}},\
  }\href@noop {} {\emph {\bibinfo {title} {Classical and Quantum
  Computation}}}\ (\bibinfo  {publisher} {American Mathematical Society},\
  \bibinfo {address} {Boston, MA, USA},\ \bibinfo {year} {2002})\BibitemShut
  {NoStop}%
\bibitem [{\citenamefont {Kitaev}(2003)}]{Kitaev2003}%
  \BibitemOpen
  \bibfield  {author} {\bibinfo {author} {\bibfnamefont {A.Yu.}\ \bibnamefont
  {Kitaev}},\ }\bibfield  {title} {\enquote {\bibinfo {title} {Fault-tolerant
  quantum computation by anyons},}\ }\href {\doibase
  http://dx.doi.org/10.1016/S0003-4916(02)00018-0} {\bibfield  {journal}
  {\bibinfo  {journal} {Ann. Phys.}\ }\textbf {\bibinfo {volume} {303}},\
  \bibinfo {pages} {2 -- 30} (\bibinfo {year} {2003})}\BibitemShut {NoStop}%
\bibitem [{\citenamefont {Bravyi}\ and\ \citenamefont
  {Kitaev}(1998)}]{Bravyi1998}%
  \BibitemOpen
  \bibfield  {author} {\bibinfo {author} {\bibfnamefont {S.~B.}\ \bibnamefont
  {Bravyi}}\ and\ \bibinfo {author} {\bibfnamefont {A.~Yu.}\ \bibnamefont
  {Kitaev}},\ }\bibfield  {title} {\enquote {\bibinfo {title} {Quantum codes on
  a lattice with boundary},}\ }\href@noop {} {\  (\bibinfo {year} {1998})},\
  \Eprint {http://arxiv.org/abs/quant-ph/9811052v1} {quant-ph/9811052v1}
  \BibitemShut {NoStop}%
\bibitem [{\citenamefont {Gottesman}(1997)}]{Gottesman1997}%
  \BibitemOpen
  \bibfield  {author} {\bibinfo {author} {\bibfnamefont {Daniel}\ \bibnamefont
  {Gottesman}},\ }\bibfield  {title} {\enquote {\bibinfo {title} {Stabilizer
  codes and quantum error correction},}\ }\href@noop {} {\bibfield  {journal}
  {\bibinfo  {journal} {Caltech Ph.D. Thesis}\ } (\bibinfo {year} {1997})},\
  \Eprint {http://arxiv.org/abs/quant-ph/9705052v1} {quant-ph/9705052v1}
  \BibitemShut {NoStop}%
\bibitem [{\citenamefont {Dennis}\ \emph {et~al.}(2002)\citenamefont {Dennis},
  \citenamefont {Kitaev}, \citenamefont {Landahl},\ and\ \citenamefont
  {Preskill}}]{Dennis2001}%
  \BibitemOpen
  \bibfield  {author} {\bibinfo {author} {\bibfnamefont {Eric}\ \bibnamefont
  {Dennis}}, \bibinfo {author} {\bibfnamefont {Alexei}\ \bibnamefont {Kitaev}},
  \bibinfo {author} {\bibfnamefont {Andrew}\ \bibnamefont {Landahl}}, \ and\
  \bibinfo {author} {\bibfnamefont {John}\ \bibnamefont {Preskill}},\
  }\bibfield  {title} {\enquote {\bibinfo {title} {Topological quantum
  memory},}\ }\href {\doibase 10.1063/1.1499754} {\bibfield  {journal}
  {\bibinfo  {journal} {J. Math. Phys.}\ }\textbf {\bibinfo {volume} {43}},\
  \bibinfo {pages} {4452--4505} (\bibinfo {year} {2002})},\ \Eprint
  {http://arxiv.org/abs/quant-ph/0110143v1} {quant-ph/0110143v1} \BibitemShut
  {NoStop}%
\bibitem [{\citenamefont {Levin}\ and\ \citenamefont {Wen}(2005)}]{Levin2005}%
  \BibitemOpen
  \bibfield  {author} {\bibinfo {author} {\bibfnamefont {Michael~A.}\
  \bibnamefont {Levin}}\ and\ \bibinfo {author} {\bibfnamefont {Xiao-Gang}\
  \bibnamefont {Wen}},\ }\bibfield  {title} {\enquote {\bibinfo {title}
  {String-net condensation: A physical mechanism for topological phases},}\
  }\href {\doibase 10.1103/PhysRevB.71.045110} {\bibfield  {journal} {\bibinfo
  {journal} {Phys. Rev. B}\ }\textbf {\bibinfo {volume} {71}},\ \bibinfo
  {pages} {045110} (\bibinfo {year} {2005})}\BibitemShut {NoStop}%
\bibitem [{\citenamefont {Hastings}\ and\ \citenamefont
  {Wen}(2005)}]{Hastings2005}%
  \BibitemOpen
  \bibfield  {author} {\bibinfo {author} {\bibfnamefont {M.~B.}\ \bibnamefont
  {Hastings}}\ and\ \bibinfo {author} {\bibfnamefont {Xiao-Gang}\ \bibnamefont
  {Wen}},\ }\bibfield  {title} {\enquote {\bibinfo {title} {Quasiadiabatic
  continuation of quantum states: The stability of topological ground-state
  degeneracy and emergent gauge invariance},}\ }\href {\doibase
  10.1103/PhysRevB.72.045141} {\bibfield  {journal} {\bibinfo  {journal} {Phys.
  Rev. B}\ }\textbf {\bibinfo {volume} {72}},\ \bibinfo {pages} {045141}
  (\bibinfo {year} {2005})}\BibitemShut {NoStop}%
\bibitem [{\citenamefont {Deutsch}(1991)}]{Deutsch1991}%
  \BibitemOpen
  \bibfield  {author} {\bibinfo {author} {\bibfnamefont {J.~M.}\ \bibnamefont
  {Deutsch}},\ }\bibfield  {title} {\enquote {\bibinfo {title} {Quantum
  statistical mechanics in a closed system},}\ }\href {\doibase
  10.1103/PhysRevA.43.2046} {\bibfield  {journal} {\bibinfo  {journal} {Phys.
  Rev. A}\ }\textbf {\bibinfo {volume} {43}},\ \bibinfo {pages} {2046--2049}
  (\bibinfo {year} {1991})}\BibitemShut {NoStop}%
\bibitem [{\citenamefont {Srednicki}(1994)}]{Srednicki1994}%
  \BibitemOpen
  \bibfield  {author} {\bibinfo {author} {\bibfnamefont {Mark}\ \bibnamefont
  {Srednicki}},\ }\bibfield  {title} {\enquote {\bibinfo {title} {Chaos and
  quantum thermalization},}\ }\href {\doibase 10.1103/PhysRevE.50.888}
  {\bibfield  {journal} {\bibinfo  {journal} {Phys. Rev. E}\ }\textbf {\bibinfo
  {volume} {50}},\ \bibinfo {pages} {888--901} (\bibinfo {year}
  {1994})}\BibitemShut {NoStop}%
\bibitem [{\citenamefont {Dymarsky}\ and\ \citenamefont
  {Liu}(2017)}]{Dymarsky2017}%
  \BibitemOpen
  \bibfield  {author} {\bibinfo {author} {\bibfnamefont {Anatoly}\ \bibnamefont
  {Dymarsky}}\ and\ \bibinfo {author} {\bibfnamefont {Hong}\ \bibnamefont
  {Liu}},\ }\bibfield  {title} {\enquote {\bibinfo {title} {Canonical
  universality},}\ }\href@noop {} {\  (\bibinfo {year} {2017})},\ \Eprint
  {http://arxiv.org/abs/1702.07722v1} {1702.07722v1} \BibitemShut {NoStop}%
\end{thebibliography}%

\end{document}